\documentclass{article}

\usepackage{ifthen}
\newboolean{useArticle} 
\setboolean{useArticle}{true}

\ifthenelse{\boolean{useArticle}}{\usepackage{authblk}
	 \usepackage{fullpage}
	 \usepackage{hyperref}
	 \usepackage{multirow}
	 \usepackage[english]{babel}
	 \usepackage{amsmath,amsfonts,amssymb,amsthm}
\usepackage{subcaption}
	 \usepackage[T1]{fontenc}
     \usepackage[cref,theorems]{nikis} 
\newcommand{\instancename}[1]{\ensuremath{\mathsf{#1}}}
\newcommand{\select}[1][]{\ensuremath{\operatorname{select}_{#1}}}
\newcommand{\rank}[1][]{\ensuremath{\operatorname{rank}_{#1}}}
}{\RequirePackage{cleveref}

\DeclareMathAlphabet{\mathup}{OT1}{\familydefault}{m}{n}
\newcommand{\bsq}[1]{\text{\lq}{#1}\text{\rq}} \newcommand{\menge}[1]{\left\{#1\right\}} \newcommand{\upgauss}[1]{\left\lceil#1\right\rceil} \newcommand{\tuple}[1]{\left(#1\right)} \newcommand{\abs}[1]{\ensuremath{\left|#1\right|}} \newcommand{\intWort}[1]{\emph{\textbf{#1}}}
\newcommand{\UnaryOperator}[2][]{\ifx&#1&\ensuremath{\mathop{}\mathopen{}#2\mathopen{}}\else \ensuremath{\mathop{}\mathopen{}#2\mathopen{}\left(#1\right)}\fi }

\newcommand{\Oh}[1]{\UnaryOperator[#1]{\mathcal{O}}}
\newcommand{\oh}[1]{\UnaryOperator[#1]{o}}
\newcommand{\select}[1][]{\ensuremath{\operatorname{select}_{#1}}}
\newcommand{\rank}[1][]{\ensuremath{\operatorname{rank}_{#1}}}
\newcommand{\instancename}[1]{\ensuremath{\mathsf{#1}}} \newcommand{\block}[1]{\subparagraph*{#1}}
\DeclareRobustCommand\scases[2]{\ifmmode #1\begin{cases}#2\end{cases}
	\else \[#1\begin{cases}#2\end{cases}\]
	\fi }
\DeclareRobustCommand\sitemize[1]{\begin{itemize}#1\end{itemize}}
\newcommand{\argmax}{\operatorname{argmax}}
\newcommand{\idbb}[1]{\ensuremath{\mathbb{#1}}}
\newcommand{\N}{\idbb{N}}
}

\usepackage{stmaryrd}
\SetSymbolFont{stmry}{bold}{U}{stmry}{m}{n}
\usepackage{anyfontsize}
\usepackage[linesnumbered,ruled,vlined]{algorithm2e}
\SetKwProg{function}{Function}{}{}

\usepackage{multicol}
\usepackage{placeins} 

\usepackage{tikz}
\usetikzlibrary{matrix,backgrounds,positioning,arrows}
\usetikzlibrary{decorations.pathreplacing,fit}
\usepackage{forest}

\newcount\Length
\pgfmathsetcount{\Length}{3}

\usetikzlibrary{arrows.meta}
\tikzset{>={Latex[width=2mm,length=1mm]}}
\newcommand{\Referenz}[8][]{\RefN{#2}{#3}{#4}{#5}{#6}{#7}{#8}{#1}}

\newcommand{\RefN}[8]{\begin{pgfonlayer}{bg} 
	\node [fit=(warray-1-#3.south) (warray-1-#4.north), inner sep=0, inner xsep=5, fill=#7,opacity=0.4, rounded corners=4pt] (src) {};
\end{pgfonlayer}
\pgfmathsetcount{\Length}{#2-#1+1}

\ifthenelse{\equal{#8}{l}}{\draw [->, rounded corners, color=#7, opacity=0.5] (warray-1-#3.#6) -- +(0,#5)-| node[above]{\texttt{(#1,\the\Length)}} (warray-1-#1.#6);
}{\draw [->, rounded corners, color=#7, opacity=0.5] (warray-1-#3.#6) -- +(0,#5) node[above] {\texttt{(#1,\the\Length)}} -| (warray-1-#1.#6);
}

}
\newcommand{\NoRef}[2]{\begin{pgfonlayer}{bg} 
		\node [fit=(warray-1-#1.south) (warray-1-#2.north), inner xsep=5, inner ysep=0, fill=gray,opacity=0.3, rounded corners=4pt] (src) {};
\end{pgfonlayer}
}

\tikzstyle{array} = [matrix of nodes,font=\ttfamily, column sep=0.5\pgflinewidth, row sep=0.5mm, nodes in empty cells,
row 1/.style={nodes={scale=0.9, fill=none, minimum size=0mm, text height=0.5em}},
row 2/.style={nodes={font=\tiny,fill=none, minimum size=0mm, text height=0.33em}},
]

\newboolean{isJournal} 
\setboolean{isJournal}{false}
\newcommand{\JOP}[2][]{#2}

\ifthenelse{\boolean{isJournal}}{\def\Jvar{}
	\newcommand{\J}[2][]{#2}
	\newcommand{\JO}[2][]{#2}

}{\usepackage{etoolbox}
	\def\Jvar{}
\newcommand{\J}[2][]{\appto\Jvar{#1#2}}
	\newcommand{\JO}[2][]{#1}
}

\usepackage{tabularx}
\newcolumntype{R}[1]{>{\raggedleft\arraybackslash}p{#1}}

\usepackage{enumerate}
\usepackage{adjustbox}
\usepackage{floatrow}
\usepackage[group-separator={,}]{siunitx}

\usepackage{booktabs}

	\definecolor{solarizedBase00}{HTML}{657b83}
	\definecolor{solarizedBase3}{HTML}{FDF6E3}
	\definecolor{solarizedYellow}{HTML}{B58900}
	\definecolor{solarizedRed}{HTML}{DC322F}
	\definecolor{solarizedMagenta}{HTML}{D33682}
	\definecolor{solarizedViolet}{HTML}{6C71C4}
	\definecolor{solarizedBlue}{HTML}{268BD2}
	\definecolor{solarizedCyan}{HTML}{2AA198}
	\definecolor{solarizedGreen}{HTML}{859900}

\usepackage{listings}
\usepackage{wrapfig}

\usepackage{ruby} 
\newcounter{Rubycount}
\newcommand{\RubyReset}{\setcounter{Rubycount}{1}
}
\newcommand{\RubyCount}[1]{\ruby{\texttt{#1}}{\theRubycount}\addtocounter{Rubycount}{1}}
\makeatletter
\newcommand{\RubyEach}[1]{\@tfor\next:=#1\do{\RubyCount{\next}}}
\makeatother

\newcommand{\UnaryBracketOperator}[2][]{\ifx&#1&\ensuremath{\mathop{}\mathopen{}#2\mathopen{}}\else \ensuremath{\mathop{}\mathopen{}#2\mathopen{}\left[#1\right]}\fi }

\newcommand{\LCP}[1][]{\UnaryBracketOperator[#1]{\instancename{LCP}}}

\newcommand{\ISA}[1][]{\UnaryBracketOperator[#1]{\instancename{ISA}}}
\newcommand{\SA}[1][]{\UnaryBracketOperator[#1]{\instancename{SA}}}
\newcommand{\BWT}[1][]{\UnaryBracketOperator[#1]{\instancename{BWT}}}
\newcommand{\emptystring}{\epsilon}
\newcommand{\ST}{\instancename{ST}}
\newcommand{\bv}[1]{\ensuremath{B_{\mathup{#1}}}}

\newcommand{\child}[1][]{\UnaryOperator[#1]{\operatorname{child}}}
\newcommand{\parent}[1][]{\UnaryOperator[#1]{\operatorname{parent}}}
\newcommand{\isleaf}[1][]{\UnaryOperator[#1]{\operatorname{is-leaf}}}

\newcommand{\leafselect}  [1][]{\UnaryOperator[#1]{\operatorname{leaf-select}}}
\newcommand{\levelanc}    [1][]{\UnaryOperator[#1]{\operatorname{level-anc}}}
\newcommand{\strdepth}    [1][]{\UnaryOperator[#1]{\operatorname{str\_depth}}}

\newcommand{\lcpcomp}{lcpcomp}
\newcommand{\CPP}{C{\tt ++}}
\newcommand{\CPPElf}{C{\tt ++}14}

\newcommand{\varT}{\ensuremath{t}}
\newcommand{\varA}{\ensuremath{b}}

\newcommand{\Footnote}[2]{\footnote{#2}}

\title{Compression with the tudocomp Framework}

\author[1]{Patrick Dinklage}
\author[1]{Johannes Fischer}
\author[1]{Dominik K{\"o}ppl}
\author[1]{Marvin L{\"o}bel}
\author[2]{Kunihiko Sadakane}
\affil[1]{Department of Computer Science, TU Dortmund, Germany\\
{\tt pdinklag@gmail.com}, {\tt johannes.fischer@cs.tu-dortmund.de}, {\tt dominik.koeppl@tu-dortmund.de}, {\tt loebel.marvin@gmail.com}}
\affil[2]{Grad.\ School of Inf.\ Science and Technology, University of Tokyo, Japan\\
{\tt sada@mist.i.u-tokyo.ac.jp}}

\ifthenelse{\boolean{useArticle}}{}{\authorrunning{P.\ Dinklage, J.\ Fischer, D.\ K{\"o}ppl, M.\ L{\"o}bel, K.\ Sadakane}
\Copyright{Patrick Dinklage, Johannes Fischer, Dominik K{\"o}ppl, Marvin L{\"o}bel, and Kunihiko Sadakane}


\EventEditors{Costas S. Iliopoulos, Solon P. Pissis, Simon J. Puglisi, and Rajeev Raman}
\EventNoEds{4}
\EventLongTitle{16th International Symposium on Experimental Algorithms (SEA 2017)}
\EventShortTitle{SEA 2017}
\EventAcronym{SEA}
\EventYear{2017}
\EventDate{June 21--23, 2017}
\EventLocation{London, UK}
\EventLogo{}
\SeriesVolume{75}
\ArticleNo{12}
}

\begin{document}
\maketitle{}

\begin{abstract}
We present a framework facilitating
the implementation and comparison of text compression algorithms.
We evaluate its features by a case study on two novel compression algorithms based on the Lempel-Ziv compression schemes 
that perform well on highly repetitive texts.
\end{abstract}

\section{Introduction}
Engineering novel compression algorithms is a relevant topic, shown by
recent approaches like bc-zip~\cite{FarruggiaFV14}, Brotli~\cite{brotli}, or Zstandard\footnote{\url{https://github.com/facebook/zstd}}.
Engineers of data compression algorithms face the fact that it is cumbersome 
(a) to build a new compression program from scratch, and (b) to evaluate and benchmark a compression algorithm against other algorithms objectively. 
We present the highly modular compression framework tudocomp that addresses both problems. 
To tackle problem~(a), tudocomp contains standard techniques like VByte~\cite{WilliamsZ99}, Elias-$\gamma$/$\delta$, or Huffman coding.
To tackle problem~(b), it provides automatic testing and benchmarking against external programs and implemented standard compressors like Lempel-Ziv compressors. 
As a case study, we present the two novel compression algorithms \lcpcomp{} and LZ78U, their implementations in tudocomp, 
and their evaluations with tudocomp.
\lcpcomp{} is based on Lempel-Ziv~77, substituting greedily the longest remaining repeated substring.
LZ78U is based on Lempel-Ziv~78, with the main difference that it allows a factor to introduce multiple new characters.

\subsection{Related Work}
There are many\footnote{e.g., \url{http://www.squeezechart.com} or \url{http://www.maximumcompression.com}} compression benchmark websites
measuring compression programs on a given test corpus. 
Although the compression ratio of a novel compression program can be compared with the ratios of the programs listed on these websites,
we cannot infer which program runs faster or more memory efficiently if these programs have not been compiled and run on the same machine.
Efforts in facilitating this kind of comparison
have been made by wrapping the source code of different compression algorithms
in a single executable that benchmarks the algorithms on the same machine with the same compile flags.
Examples include lzbench\footnote{\url{https://github.com/inikep/lzbench}} and Squash\footnote{\url{https://quixdb.github.io/squash-benchmark}}.

Considering frameworks aiming at easing the comparison \emph{and} implementation of new compression algorithms, 
we are only aware of the \CPP{}98 library ExCom~\cite{HolubRS11}.
The library contains a collection of compression algorithms.
These algorithms can be used as components for a compression pipeline.
However, ExCom does not provide the same flexibility as we had in mind;
it provides only character-wise pipelines, i.e., it does no bitwise transmission of data.
Its design does not use meta-programming features;
a header-only library has more potential for optimization since the compiler can inline header-implemented (possibly performance critical) functions easily.

A broader focus is set in Giuseppe Ottaviano's succinct library~\cite{GrossiO13}
and Simon Gog's Succinct Data Structure Library 2.0 (SDSL)~\cite{gbmp2014sea}.
These two libraries provide integer coders and helper functions for working on the bit level.

\subsection{Our Results/Approach}
Our lossless compression framework \intWort{tudocomp} aims at supporting and facilitating the implementation of novel compression algorithms.
The philosophy behind tudocomp is to support building a pipeline of modules that transforms an input to a compressed binary output.
This pipeline has to be flexible:
appending, exchanging and removing a module in the pipeline in a plug-and-play manner is in the main focus of the design of tudocomp.
Even a module itself can be refined into submodules. 

To this end, tudocomp is written in modern \CPPElf{}.
On the one hand, the language allows us to write compile time optimized code due to its meta programming paradigm.
On the other hand, its fine-grained memory management mechanisms support controlling and monitoring the memory footprint in detail.
We provide a tutorial, an exhaustive documentation of the API, and the source code at \url{http://tudocomp.org} with the permissive Apache License~2.0 to encourage developers to use and foster the framework.

In order to demonstrate its usefulness, we added reference implementations of common compression and encoding schemes (see \Cref{secDescrFramework}).
On top of that, we present two novel algorithms (see \Cref{secNewAlgo}) which we have implemented in our framework.
We give a detailed evaluation of these algorithms in \Cref{secEvaluation},
thereby exposing the benchmarking and the visualization tools of tudocomp.

\section{Description of the tudocomp Framework}\label{secDescrFramework}
On the topmost abstraction level, tudocomp defines the abstract types {\tt Compressor} and {\tt Coder}.
A \intWort{compressor} transforms an input into an output so that the input can be losslessly restored from the output by the corresponding \intWort{decompressor}. 
A \intWort{coder} takes an elementary data type like a character and writes it to a compressed bit sequence.
As with compressors, each coder is accompanied by a \intWort{decoder} taking care of restoring the original data from its compressed bit sequence.
By design, a coder can take the role of a compressor, but a compressor may not be suitable as a coder (e.g., a compressor that needs random access on the whole input).

tudocomp provides implementations of the compressors and the coders shown in the tables below.
Each compressor and coder gets an identifier (right column of each table).
\begin{multicols}{2}
\hspace{-2em}
\begin{tabular}{ll}
\multicolumn{2}{l}{Compressors}\\
	\toprule
BWT               & {\tt bwt}
\\
Coder wrapper    & {\tt encode}
\\
LCPComp~(\Cref{seclcpcomp})               & {\tt lcpcomp}
\\
LZ77~(Def.~\ref{defLZ77}), LZSS~\cite{StorerS82} output   & {\tt lzss\_lcp}
\\
LZ78~(Def.~\ref{defLZ78})              & {\tt lz78}
\\
LZ78U~(\Cref{secLZ78U})             & {\tt lz78u}
\\
LZW~\cite{Welch84}               & {\tt lzw}
\\
Move-To-Front               & {\tt mtf}
\\
Re-Pair~\cite{LarssonM99}            & {\tt repair}
\\
Run-Length-Encoding         & {\tt rle}
\\
\bottomrule
\end{tabular}

\hspace{-2em}
\begin{tabular}{ll}
	\multicolumn{2}{l}{Integer Coders} \\
 	\toprule
Bit-Compact Coder & {\tt bit}
\\
Elias-$\gamma$~\cite{elias75universal} & {\tt gamma}
\\
Elias-$\delta$~\cite{elias75universal} & {\tt delta} 
\\
\bottomrule
\end{tabular}

\hspace{-2em}
\begin{tabular}{p{5cm}l}
	\multicolumn{2}{l}{String Coders}\\
 	\toprule
 Canonical Huffman Coder~\cite{witten99managing} & {\tt huff}
 \\
A Custom Static Low Entropy Encoder (Section \ref{seclcpcomp}) & {\tt sle}
 \\
 \bottomrule
\end{tabular}
\vspace{-20em}

\end{multicols}

The behavior of a compressor or coder can be modified by passing different parameters.
A parameter can be an elementary data type like an integer, but it can also be an instance of a class that specifies certain subtasks
like integer coding.
For instance, the compressor~{\tt lzss\_lcp(threshold, coder)} takes an integer {\tt threshold}  and a {\tt coder} (to code an LZ77 factor) as parameters.
The coder is supplied as a parameter such that the compressor can call the coder directly (instead of alternatively piping the output of {\tt lzss\_lcp} to a coder).

The support of class parameters eases the deployment of the design pattern \intWort{strategy}~\cite{Gamma1995}.
A strategy determines what algorithm or data structure is used to achieve a compressor-specific task.

\block{Library and Command Line Tool}
tudocomp consists of two major components: 
a standalone compression library and a command line tool {\tt tdc}.
The library contains the core interfaces and implementations of the aforementioned compressors and coders.
The tool~{\tt tdc} exposes the library's functionality in form of an executable that can run compressors directly on the command line.
It allows the user to select a compressor by its identifier and to pass parameters to it, i.e., the user can specify the exact compression strategy \emph{at runtime}.

\block{Example}
For instance, the LZ78U compressor (\Cref{secLZ78U}) expects a compression strategy, an integer coder, and an integer variable specifying a threshold.
Its strategy can define parameters by itself, like which string coder to use.
A valid call is \texttt{./tdc -a 'lz78u(coder = bit, comp = buffering(string\_coder = huff), threshold = 3)' input.txt -o\linebreak[4] output.tdc},
where {\tt tdc} compresses the file {\tt input.txt} and stores the compressed bit sequence in the file {\tt output.tdc}.
To this end, it uses the compressor {\tt lz78u} parametrized by the coder {\tt bit} for integer values, by the compression strategy {\tt buffering} with {\tt huff} to code strings, 
and by a threshold value of~$3$.
Note that {\tt coder} and {\tt string\_coder} are parameters for two independently selectable coders.
When selecting a coder we have to pay attention that 
a static entropy coder like {\tt huff} needs to parse its input in advance (to generate a codeword for each occurring character).
To this end, we can only apply the coder {\tt huff} with a compression strategy that buffers the output (for {\tt lz78u} this strategy is called {\tt buffering}).
To stream the output (i.e., the opposite of buffering the complete output in RAM), 
we can use the alternative strategy {\tt streaming}. This strategy also requires a coder, but contrary to the buffering strategy, that coder does not need to look at the complete output (e.g., universal codes like {\tt gamma}).

In this fashion, we can build more sophisticated compression pipelines like {\tt lzma} applying different coders for literals, pointers, and lengths.
Each coder is unaware of the other coders, as if every coder was processing an independent stream.

\block{Decompression}
After compressing an input using a certain compression strategy, the tool adds a header to the compressed file so that
it can decompress it without the need for specifying the compression strategy again.
However, this behavior can be overruled by explicitly specifying a decompression strategy, e.g., in order to test 
different decompression strategies.

\begin{figure}[t]
	\centering{\includegraphics[width=\textwidth]{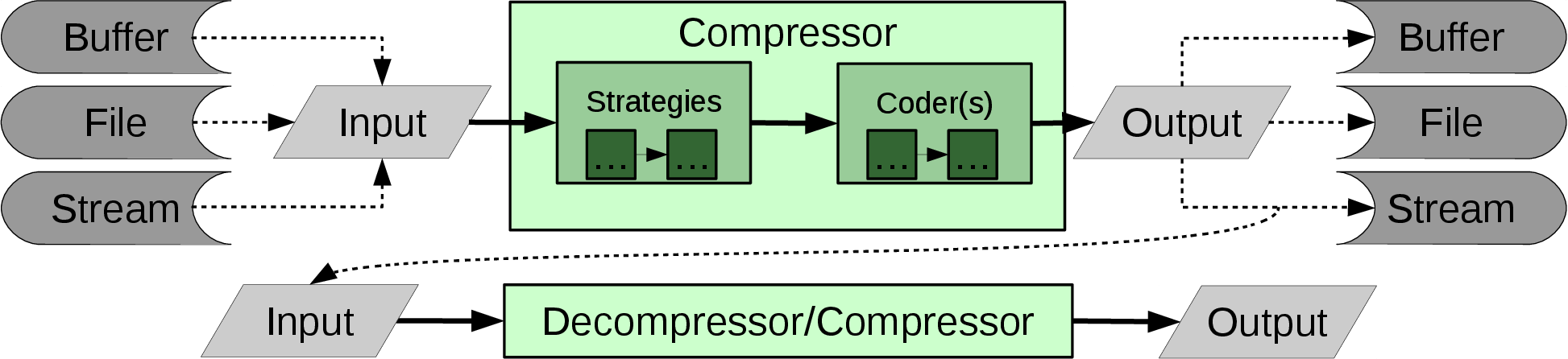}
	}
	\caption{Flowchart of a possible compression pipeline.
The compressors of tudocomp work with abstract data types for input and output, i.e.,
a compressor is unaware of whether the input or the output is a file, is stored in memory, or is accessed using a stream.
A compressor can follow one or more compression strategies that can have (nested) parameters.
Usually, a compressor is parametrized with one or more coders (e.g., for different integer ranges or strings) that produce the final output.
	}
	\label{figFlowchart}
\end{figure}

\block{Helper classes}
tudocomp provides several classes for easing common tasks when engineering a new compression algorithm,
like the computation of \SA{}, \ISA{} or \LCP{}.
tudocomp generates \SA{} with {\tt divsufsort}\Footnote{divsufsort}{\url{https://github.com/y-256/libdivsufsort}},
and \LCP{} with the $\Phi$-algorithm~\cite{kaerkkaeinen09permuted}.
The arrays \SA{}, \ISA{}, and \LCP{} can be stored in plain arrays or in packed arrays with a bit width of $\upgauss{\lg n}$ (where $n$ is the length of the input text), 
i.e., in a \intWort{bit-compact representation}.
We provide the modes {\tt plain}, {\tt compressed}, and {\tt delayed} to describe when/whether a data structure should be 
stored in a bit-compact representation: In {\tt plain} mode, all data structures are stored in plain arrays;
in {\tt compressed} mode, all data structures are built in a bit-compact representation.
In {\tt delayed} mode, tudocomp first builds a data structure~$A$ in a plain array; 
when all other data structures are built whose constructions depended on~$A$, $A$ gets transformed into a bit-compact representation.
\JOP{While {\tt direct} and {\tt compressed} are the fastest or the memory-friendliest modes, respectively, 
the data structures produced by {\tt delayed} are the same as {\tt compressed}, though {\tt delayed} is faster than {\tt compressed}.
}

If more elaborated algorithms are desired (e.g., for producing compressed data structures like the compressed suffix array),
it is easy to use tudocomp in conjunction with SDSL
for which we provide an easy binding.

\block{Combining streaming and offline approaches}
A compressor can stream its input (online approach) or request the input to be loaded into memory (offline approach).
Compressors can be chained to build a pipeline of multiple compression modules, like as in \Cref{figFlowchart}.

\subsection{Example Implementation of a Compressor}\label{secExampleBWT}
\begin{wrapfigure}{l}{.5\textwidth}
	\vspace{-2em}
	\begin{flushleft}
		{\scriptsize (a) \CPP{} Source Code}
	\end{flushleft}
	\vspace{-1.5em}
\begin{lstlisting}
#include <tudocomp/tudocomp.hpp>
class BWTComp : public Compressor {
  public: static Meta meta() {
    Meta m("compressor", "bwt");
    m.option("ds").templated<TextDS<>>();
    m.needs_sentinel_terminator();
    return m; }
  using Compressor::Compressor;
  void compress(Input& in, Output& out) {
    auto o = out.as_stream();
    auto i = in.as_view();
    TextDS<> t(env().env_for_option("ds"),i);
    const auto& sa = t.require_sa();
    for(size_t j = 0; j < t.size(); ++j)
      o << ((sa[j] != 0) ? t[sa[j] - 1]
                         : t[t.size() - 1]);
  }
  void decompress(Input&, Output&){/*[...]*/}
};
\end{lstlisting}
	\vspace{-2em}
	\begin{flushleft}
{\scriptsize (b) Execution with {\tt tdc} }
	\end{flushleft}
	\vspace{-1.5em}
\begin{lstlisting}[language=bash, morekeywords={hexdump,stat},emph={tdc},emphstyle={\color{solarizedViolet}}]
> echo -n 'aaababaaabaababa' > ex.txt
> ./tdc -a bwt -o bwt.tdc ex.txt
> hexdump -v -e '"b w t > ./tdc -a 'bwt:rle' -o rle.tdc ex.txt
> hexdump -v -e '"b  w  t  :  r  l  e  \end{lstlisting}
\vspace{-1.5em}
\end{wrapfigure}
The source code (a) on the left implements a compressor that computes the Burrows-Wheeler transform~(BWT) (see \Cref{secTheo}) of an input.
To this end, it loads the input into memory using (line 11) {\tt in.as\_view()} and computes the suffix array using (line 13) {\tt t.require\_sa()}.
In the function {\tt meta}, we state that we assume the unique terminal symbol 
(represented by the byte \bsq{\tt\textbackslash{}0}) as part of the text, and that
we want to register the class {\tt BWTComp} as a {\tt Compressor} with the identifier {\tt bwt}.
By doing so, we can call the compressor directly in the command line tool~{\tt tdc} using the argument {\tt -a bwt}.
In the shell code (b) on the left, you can see how we produced the BWT of our running example.
The program {\tt hexdump} outputs each character of a file such that non-visible characters are escaped.
A {\tt \%}-sign separates the header from the body in the output.
Next, we use the binary composition operator~{\tt :} connecting the output of its left operand with the input of its right operand.
In the shell code, this operator pipes the output of {\tt bwt} to the run-length encoding compressor {\tt rle}, which
transforms a substring $\texttt{aa}\hspace{-0.3em}\underbrace{\texttt{a}\cdots\texttt{a}}_{m~\textrm{times}}$ to ${\tt aa}m$ with 
$m \ge 0$ encoded in VByte (the output is a \emph{byte} sequence).

\begin{wrapfigure}{l}{.5\textwidth}
	\vspace{-2em}
	\begin{flushleft}
{\scriptsize (c) Assembling a compression pipeline }
	\end{flushleft}
	\vspace{-1.5em}
\begin{lstlisting}[language=bash, morekeywords={hexdump,stat},emph={tdc},emphstyle={\color{solarizedViolet}}]
> ./tdc -a bwt -o bwt.tdc pc_english.200MB
> ./tdc -a 'bwt:rle:mtf:encode(huff)' -o bzip.tdc pc_english.200MB
> stat -c"209715200 pc_english.200MB
209715209 bwt.tdc
66912437 bzip.tdc
\end{lstlisting}
\JO[\vspace{-2.5em}]{\vspace{-1em}}
\end{wrapfigure}
Finally, the compressor {\tt bwt} can be used as part of a pipeline to achieve good compression quality:
Given a move-to-front compressor {\tt mtf} and a Huffman coder {\tt huff}, we can build a chain
{\tt bwt:rle:mtf:encode(huff)}.
The compressor {\tt encode} is a wrapper that turns a coder into a compressor.
The last code fragment~(c) on the left shows the calls of this pipeline and a call of {\tt bwt} only.
Using {\tt stat}, we measure the file sizes (in bytes) of the input \textsc{pc-english} (see \Cref{secEvaluation}) and both outputs.

\subsection{Specific Features}
tudocomp excels with the following additional properties:

\block{Few Build Requirements}
To deploy tudocomp, the build management software {\tt cmake}, the version control system {\tt git}, Python 3, and a \CPPElf{} compiler are required.
{\tt cmake} automatically downloads and builds other third-party software components like the SDSL\@.
We tested the build process on Unix-like build environments, namely
Debian Jessie, Ubuntu Xenial, Arch Linux 2016, and the Ubuntu shell on Windows 10. 

\block{Unit Tests}
tudocomp offers semi-automatic unit tests.
For a registered compressor, tudocomp can automatically generate test cases that check 
whether the compressor can compress and decompress a set of selected inputs successfully. 
These inputs include border cases like the empty string, a run of the same character, samples on various subranges in UTF-8,
Fibonacci strings, Thue-Morse strings, and strings with a high number of runs~\cite{runrich}.
These strings can be generated on-the-fly by {\tt tdc} as an alternative input.

\block{Type Inferences}
The \CPP{} standard does neither provide a syntax for constraining type parameters (like generic type bounding in Java)
nor for querying properties of a class at runtime (i.e., reflection).
To address this syntactic lack,
we augment each class exposed to {\tt tdc} and to the unit tests with a so-called \intWort{type}.
A type is a string identifier. 
We expect that classes with the same type provide the same public methods.
Types resemble \emph{interfaces} of Java, but contrary to those, they are not subject to polymorphism.
Common types in our framework are {\tt Compressor} and {\tt Coder}.
The idea is that, given a compressor that accepts a {\tt Coder} as a parameter, it should accept all classes of type {\tt Coder}.
To this end, each typed class is augmented with an identifier and a description of all parameters that the class accepts.
All typed classes are exposed by the tool {\tt tdc} that calls a typed class by its identifier with the described parameters.
Types provide a uniform, but simple declaration of all parameters (e.g., integer values, or strategy classes).
The aforementioned exemplaric call of {\tt lz78u} at the beginning of \Cref{secDescrFramework} illustrates the uniform declaration of the parameters of a compressor.

\block{Evaluation tools}
To evaluate a compressor pipeline, tudocomp provides several tools to facilitate 
measuring the compression ratio, the running time, and the memory consumption.
By adding {\tt -{}-stats} to the parameters of {\tt tdc}, the tool monitors these measurement parameters:
It additionally tracks the running time and the memory consumption of the data structures in all phases.
A phase is a self-defined code division like a pre-processing phase, or an encoding phase.
Each phase can collect self-defined statistics like the number of generated factors.
All measured data is collected in a JSON file that can be visualized by the web application found at \url{http://tudocomp.org/charter}.
An example is given in \Cref{figEnglishDiagram}.

In addition, we have a command line \intWort{comparison tool} called {\tt compare.py} that runs a predefined set of compression programs 
(that can be tudocomp compressors or external compression programs). 
Its primary usage is to compare tudocomp compression algorithms with external compression programs.
It monitors the memory usage with the tool {\tt valgrind --tool=massif --pages-as-heap=yes}.
This tool is significantly slower than running {\tt tdc} with {\tt -{}-stats}.

\section{New Compression Algorithms}\label{secNewAlgo}
With the aid of tudocomp, it is easy to implement new compression algorithms.
We demonstrate this by introducing two novel compression algorithms: \intWort{\lcpcomp{}} and \intWort{LZ78U}.
To this end, we first recall some definitions.

\subsection{Theoretical Background}\label{secTheo}
Let $\Sigma$ denote an integer alphabet of size $\sigma = \abs{\Sigma} \le n^{\Oh{1}}$ for a natural number~$n$.
We call an element $T \in \Sigma^*$ a \intWort{string}.
The empty string is $\emptystring$ with $\abs{\emptystring}=0$.
Given $x,y,z \in \Sigma^*$ with $T = xyz$, 
then $x$, $y$ and $z$ are called a \intWort{prefix}, \intWort{substring} and \intWort{suffix} of $T$, respectively.
We call $T[i..]$ the $i$-th suffix of $T$, and denote a substring $T[i] \cdots T[j]$ with $T[i..j]$.

For the rest of the article, we take a string~$T$ of length~$n$.
We assume that $T[n]$ is a special character $\texttt{\$} \notin \Sigma$ smaller than all characters of $\Sigma$ so that no suffix of $T$ is a prefix of another suffix of $T$.

\SA{} and \ISA{} denote the suffix array~\cite{manber93suffix} and the inverse suffix array of $T$, respectively.
\LCP[2..n] is an array such that \LCP[i] is the length of the longest common prefix of the \emph{lexicographically} $i$-th smallest suffix with its lexicographic predecessor for $i=2,\ldots,n$.
The BWT~\cite{burrows94block} of~$T$ is the string~$\BWT$ 
with 
	\scases{\BWT[j] =}{T[n] & \text{if~} {\SA[j]} = 1, \\
		T[{\SA[j]-1}] & \text{otherwise,\ }
}for $1 \le j \le n$.
The arrays \SA{}, \ISA{}, \LCP{} and \BWT{} can be constructed in time linear to the number of characters of~$T$~\cite{kaerkkaeinen06linear}.

As a running example, we take the text $T := {\tt aaababaaabaababa\$} $.
The arrays \SA{}, \LCP{} and \BWT{} of this example text are shown in \Cref{figSALCP}.

\begin{figure}[t]
\floatbox[{\capbeside\thisfloatsetup{capbesideposition={right},capbesidewidth=3cm}}]{figure}[\FBwidth]
	{\caption{Suffix array, inverse suffix array, LCP array and BWT of the running example.}
	\label{figSALCP}
}{\tabcolsep=0.1em
	\begin{tabular}{l*{17}{R{1.3em}}}
$i$   & 1  & 2  & 3 & 4 & 5 & 6  & 7 & 8  & 9 & 10 & 11 & 12 & 13 & 14 & 15 & 16 & 17
			\\\midrule
			$T$   & {\tt a}  & {\tt a}  & {\tt a} & {\tt b} & {\tt a} & {\tt b}  & {\tt a} & {\tt a}  & {\tt a} & {\tt b}  & {\tt a}  & {\tt a}  & {\tt b}  & {\tt a}  & {\tt b}  & {\tt a}  & {\tt \$}
\\
{\SA[i]}  & 17 & 16 & 7 & 1 & 8 & 11 & 2 & 14 & 5 & 9  & 12 & 3  & 15 & 6  & 10 & 13 & 4
\\
{\ISA[i]} &4&7&12&17&9&14&3&5&10&15&6&11&16&8&13&2&1
\\
{\LCP[i]} & -  & 0 & 1 & 5 & 2 & 4 & 6 & 1 & 3 & 4 & 3 & 5 & 0 & 2 & 3 & 2 & 4
\\
{\BWT[i]} & {\tt a} & {\tt b} & {\tt b} & {\tt \$} & {\tt a} & {\tt b} & {\tt a} & {\tt b} & {\tt b} & {\tt a} & {\tt a} & {\tt a} & {\tt a} & {\tt a} & {\tt a} & {\tt a} & {\tt a}
		\end{tabular}
}\end{figure}

Given a bit vector $\bv{}$ with length $\abs{\bv{}}$, 
the operation $\bv{}.\rank[1](i)$ counts the number of \bsq{1}-bits in $\bv{}[1..i]$, 
and the operation $\bv{}.\select[1](i)$ yields the position of the $i$-th \bsq{1} in $\bv{}$.

There are data structures~\cite{jacobson89space,clark96compact} that can answer $\rank{}$ and $\select{}$ queries on $\bv{}$ in constant time, respectively.
Each of them uses $\oh{\abs{\bv{}}}$ additional bits of space, and both can be built in $\Oh{\abs{\bv{}}}$ time.

The \intWort{suffix trie} of $T$ is the trie of all suffixes of $T$.
The \intWort{suffix tree}~\cite{weiner73linear} of $T$, denoted by \ST{}, is the tree obtained by compacting the suffix trie of $T$.
It has $n$ leaves and at most $n$ internal nodes.
The string stored in an edge~$e$ is called the \intWort{edge label} of $e$, and denoted by~$\lambda(e)$.
The \intWort{string depth} of a node $v$ is the length of the concatenation of all edge labels on the path from the root to $v$.
The leaf corresponding to the $i$-th suffix is labeled with $i$.

Each node of the suffix tree is uniquely identified by its pre-order number.
We can store the suffix tree topology in a bit vector (e.g., DFUDS~\cite{Benoit2005RTo} or BP~\cite{jacobson89space,cstsada})
such that \rank{} and \select{} queries enable us to address a node by its pre-order number in constant time.
If the context is clear, we implicitly convert an \ST{} node to its pre-order number, and vice versa.
We will use the following constant time operations on the suffix tree:
\begin{itemize}
	\item  $\parent[v]$ selects the parent of the node~$v$,
	\item  $\levelanc[\ell,d]$ selects the ancestor of the leaf~$\ell$ at depth~$d$ (level ancestor query), and
	\item  $\leafselect[i]$ selects the $i$-th leaf (in lexicographic order).
\end{itemize}

A \intWort{factorization} of $T$ of size~$z$ partitions $T$ into $z$~substrings $T=F_1 \cdots F_z$.
These substrings are called \intWort{factors}. In particular, we have:

\begin{definition}\label{defLZ77}
	A factorization $F_1\cdots F_z = T$ is called the \intWort{Lempel-Ziv-77 (LZ77) factorization}~\cite{ziv77universal} of $T$ with a threshold~$\varT \ge 1$ iff 
$F_x$ is either the longest substring of length at least~$\varT$ occurring at least twice in $F_1 \cdots F_x$, or, if such a substring does not exist, 
a single character. 
We merge successive occurrences of the latter type of factors to a single factor and call it a \intWort{remaining substring}.
\end{definition}
The usual definition of the LZ77 factorization fixes $\varT=1$.
We introduced the version with a threshold to make the comparison with \lcpcomp{} (\Cref{seclcpcomp}) fairer.

\begin{definition}\label{defLZ78}
	A factorization $F_1\cdots F_z = T$ is called the \intWort{Lempel-Ziv-78 (LZ78) factorization}~\cite{ziv78compression} of $T$ iff 
$F_x=F_y\cdot c$ with $F_y = \argmax_{S \in \menge{F_y : y < x} \cup \menge{\epsilon} } \abs{S}$ and $c\in\Sigma$
for all $1 \le x \le z$.
\end{definition}

\subsection{\lcpcomp{}}\label{seclcpcomp}
The idea of \lcpcomp{} is to search for long repeated substrings and
substitute one of their occurrences with a reference to the other.
Large values in the LCP-array indicate such long repeated substrings.
There are two major differences to the LZ77 compression scheme:
(1) while LZ77 only allows back-references, \lcpcomp{} allows both back \emph{and forward} references; and (2) LZ77 factorizes $T$ greedily from left to right, whereas \lcpcomp{} makes substitutions at \emph{arbitrary} positions in the text, greedily chosen such that the number of substituted characters is maximized.
This process is repeated until all remaining repeated substrings are shorter than a threshold~$\varT$.
On termination, \lcpcomp{} has generated a factorization $T = F_1 \cdots F_z$, 
where each $F_j$ is either a remaining substring, or a
reference~$(i,\ell)$ with the intended meaning ``copy $\ell$ characters from position $i$''
(see \Cref{figParsingArrowsLCPcomp} for an example).

\begin{figure}[t]
\begin{minipage}[b]{.5\textwidth}
\begin{adjustbox}{trim=0 0cm 0 0.5cm}
\centering{\begin{tikzpicture}[transform shape, inner sep=1mm ]
		\pgfdeclarelayer{bg}    \pgfsetlayers{bg,main}  \matrix[array] (warray) {a \pgfmatrixnextcell a \pgfmatrixnextcell a \pgfmatrixnextcell b \pgfmatrixnextcell a \pgfmatrixnextcell b \pgfmatrixnextcell a \pgfmatrixnextcell a \pgfmatrixnextcell a \pgfmatrixnextcell b \pgfmatrixnextcell a \pgfmatrixnextcell a \pgfmatrixnextcell b \pgfmatrixnextcell a \pgfmatrixnextcell b \pgfmatrixnextcell a \pgfmatrixnextcell \$ \\
 1\pgfmatrixnextcell 2\pgfmatrixnextcell 3\pgfmatrixnextcell 4\pgfmatrixnextcell 5\pgfmatrixnextcell 6\pgfmatrixnextcell 7\pgfmatrixnextcell 8\pgfmatrixnextcell 9\pgfmatrixnextcell 10\pgfmatrixnextcell 11\pgfmatrixnextcell 12\pgfmatrixnextcell 13\pgfmatrixnextcell 14\pgfmatrixnextcell 15\pgfmatrixnextcell 16\pgfmatrixnextcell 17
\\};
{\NoRef{1}{1}}{\Referenz{1}{2}{2}{3}{0.2}{north}{solarizedYellow}}{\NoRef{4}{4}}{\Referenz{3}{5}{5}{7}{0.2}{north}{solarizedRed}}{\Referenz{2}{5}{8}{11}{0.25}{north}{solarizedGreen}}{\Referenz{3}{7}{12}{16}{0.29}{north}{solarizedBlue}}{\NoRef{17}{17}}\end{tikzpicture}
}\end{adjustbox}
\subcaption{LZ77}\label{figParsingArrows77}
\end{minipage}\begin{minipage}[b]{.5\textwidth}
\begin{adjustbox}{trim=0 0cm 0 0.5cm}
\centering{\begin{tikzpicture}[transform shape, inner sep=1mm ]
		\pgfdeclarelayer{bg}    \pgfsetlayers{bg,main}  \matrix[array] (warray) {a \pgfmatrixnextcell a \pgfmatrixnextcell a \pgfmatrixnextcell b \pgfmatrixnextcell a \pgfmatrixnextcell b \pgfmatrixnextcell a \pgfmatrixnextcell a \pgfmatrixnextcell a \pgfmatrixnextcell b \pgfmatrixnextcell a \pgfmatrixnextcell a \pgfmatrixnextcell b \pgfmatrixnextcell a \pgfmatrixnextcell b \pgfmatrixnextcell a \pgfmatrixnextcell \$ \\
 1\pgfmatrixnextcell 2\pgfmatrixnextcell 3\pgfmatrixnextcell 4\pgfmatrixnextcell 5\pgfmatrixnextcell 6\pgfmatrixnextcell 7\pgfmatrixnextcell 8\pgfmatrixnextcell 9\pgfmatrixnextcell 10\pgfmatrixnextcell 11\pgfmatrixnextcell 12\pgfmatrixnextcell 13\pgfmatrixnextcell 14\pgfmatrixnextcell 15\pgfmatrixnextcell 16\pgfmatrixnextcell 17
\\};
{\NoRef{1}{1}}
{\Ref[below]{11}{16}{2}{7}{0.2}{north}{solarizedRed}}{\NoRef{8}{8}}
{\Ref[l]{5}{6}{9}{10}{0.4}{north}{solarizedGreen}}{\Referenz{8}{11}{11}{14}{0.3}{north}{solarizedBlue}}{\NoRef{15}{17}}\end{tikzpicture}
}\end{adjustbox}
\subcaption{\lcpcomp{}}\label{figParsingArrowsLCPcomp}
\end{minipage}
\caption{References of the (a) LZ77 factorization 
	with the threshold~$\varT = 2$, and of the 
	(b) \lcpcomp{} factorization
		with the same threshold.
		The output of the LZ77 and the \lcpcomp{} algorithms are
			{\tt a}{\tt (1,2)}{\tt b}{\tt (3,3)}{\tt (2,4)}{\tt (3,5)}{\tt \$} 
			and
			{\tt a}{\tt (11,6)}{\tt a}{\tt (5,2)}{\tt (8,4)}{\tt ba\$}, respectively.
		}
		\label{figParsingArrows}
\end{figure}

\block{Algorithm}
The LCP array stores the longest common prefix of two lexicographically neighboring suffixes.
The largest entries in the LCP array correspond to the longest substrings of the text that have at least two occurrences.
Given a suffix~$T[\SA[i]..]$ whose entry \LCP[i] is maximal among all other values in \LCP,
we know that $T[\SA[i]..\SA[i]+\LCP[i]-1] = T[\SA[i-1]..\SA[i-1]+\LCP[i]-1]$, i.e., we can substitute~$T[\SA[i]..\SA[i]+\LCP[i]-1]$ with the reference~$(\SA[i-1],\LCP[i])$.
In order to find a suffix whose \LCP{} entry is maximal, we need a data structure that 
maintains suffixes ordered by their corresponding LCP values.
We use a maximum heap for this task. 
To this end, the heap stores suffix array indices whose keys are their LCP values (i.e., insert $i$ with key \LCP[i], $2 \le i \le n$).
The heap stores only those indices whose keys are at least~$\varT$.

While the heap is not empty, we do the following:
\begin{enumerate}
	\item Remove the maximum from the heap; let $i$ be its value.
	\item Report the reference $(\SA[i-1], \LCP[i])$ and the position \SA[i] as a triplet $(\SA[i-1],$ $ \LCP[i],$ $ \SA[i])$.
	\item For every $1 \le k \le \LCP[i]-1$, remove the entry $\ISA[ {\SA[i]+k} ]$ from the heap (as these positions are covered by the reported reference).
	\item Decrease the keys of all entries~$j$ with $\SA[i]-\LCP[i] \le \SA[j] < \SA[i]$ to $\min(\LCP[j], \SA[i]-\SA[j])$.
		(If a key becomes smaller than $\varT$, remove the element from the heap.)
		By doing so, we prevent the substitution of a substring of~$T[\SA[i]..\SA[i]+\LCP[i]-1]$ at a later time.
\end{enumerate}

\begin{wrapfigure}{l}{.5\textwidth}
\begin{lstlisting}
template<class text_t>
class MaxHeapStrategy : public Algorithm {
 public: static Meta meta() {
  Meta m("lcpcomp_strategy", "heap");
  return m; }
 using Algorithm::Algorithm;
 void create_factor(size_t pos, size_t src, size_t len);
 void factorize(text_t& text, size_t t) {
  text.require(text_t::SA | text_t::ISA | text_t::LCP);
  auto& sa = text.require_sa();
  auto& isa = text.require_isa();
  auto lcpp = text.release_lcp()->relinquish();
  auto& lcp = *lcpp;
  ArrayMaxHeap<typename text_t::lcp_type::data_type> heap(lcp, lcp.size(), lcp.size());
  for(size_t i = 1; i < lcp.size(); ++i)
   if(lcp[i] >= t) heap.insert(i);
  while(heap.size() > 0) {
   size_t i = heap.top(), fpos = sa[i],
       fsrc = sa[i-1], flen = heap.key(i);
   create_factor(fpos, fsrc, flen);
   for(size_t k=0; k < flen; k++)
    heap.remove(isa[fpos + k]);
   for(size_t k=0;k < flen && fpos > k;k++) {
    size_t s = fpos - k - 1;
    size_t j = isa[s];
    if(heap.contains(j)) {
     if(s + lcp[j] > fpos) {
      size_t l = fpos - s;
      if(l >= t)
       heap.decrease_key(j, l);
      else heap.remove(j);
}}}}}};
\end{lstlisting}
\JO[\vspace{-2em}]{\vspace{-5em}}
\end{wrapfigure}

As an invariant, the key~$\ell$ of a suffix array index~$i$ stored in the heap will always be the maximal number of characters such that~$T[i..i+\ell-1]$ occurs at least twice in the \emph{remaining} text.

The reported triplets are collected in a list.
To compute the final output, we sort the triplets by their third component (storing the starting position of the substring substituted by the reference stored in the first two components).
We then scan simultaneously over the list and the text to generate the output. 
\Cref{figElaboratedLcpcompExample} demonstrates how the \lcpcomp{} factorization of the running example is done step-by-step.

The code on the left implements the compression strategy of \lcpcomp{} that uses a maximum heap.
We transfered the code from the compressor class to a strategy class since the \lcpcomp{} compression scheme can be implemented in different ways.
Each strategy receives a text. Its goal is to compute all factors (created by the {\tt create\_factor} method).
In the depicted strategy, we use a maximum heap to find all factors.
The heap is implemented in the class {\tt ArrayMaxHeap}. 
An instance of that class stores an array~$A$ of keys and an array heap maintaining (key-value)-pairs of the form~$(A[i],i)$ with the order
$(A[i],i) < (A[j],j) :\Leftrightarrow A[i] < A[j]$. 
To access a specific element in the heap by its value, the class has an additional array storing the position of each value in the heap.

\JO[Although a reference~$r$ can refer to a substring that has been substituted by another reference after the creation of~$r$,
in Lem.~\ref{lemNoCycles} (Appendix), we show that it is always possible to restore the text.
]{\block{Correctness} Although a reference~$r$ can refer to a substring that has been substituted by another reference after the creation of~$r$, it is still possible to restore the text due to the following lemma:
}\J[\section{Cycle-Free Lemma of lcpcomp}
]{\begin{lemma}\label{lemNoCycles}
	The output of \lcpcomp{} contains enough information to restore the original text.
\end{lemma}
\begin{proof}
	We want to show that the output is free of cycles, i.e.,
	there is no text position~$i$ for that $i \underbrace{\rightarrow \cdots \rightarrow}_{\text{cycle length}} i$ holds,
	where $\rightarrow$ is a relation on text positions such that
	$i \rightarrow j$ holds iff there is a substring $T[i'..i'+\ell-1]$ with $i \in [i',i'+\ell-1]$ that has been substituted by a reference~$(j-i+i',\ell)$.
	If the text is free of cycles, then each substituted text position can be restored by following a finite chain of references.
	
	First, we show that is not possible to create cycles of length two.
	Assume that we substituted $T[\SA[i]..\SA[i]+\ell_i-1]$ with $(\SA[i-1],\ell_i)$ for $\varT \le \ell_i \le \LCP[i]$.
	The algorithm will not choose
	$T[\SA[i-1]+k..\SA[i-1]+k+\ell_k-1]$ for $0 \le k \le \ell_i$ and $\varT \le \ell_k \le \LCP[i]-k$ to be substituted with $(\SA[i]+k,\ell_k)$,
	since $T[\SA[i]+k..] > T[\SA[i-1]+k..]$ and therefore $\ISA[ {\SA[i]+k} ] > \ISA[ {\SA[i-1]+k} ]$.
	Finally, by the transitivity of the lexicographic order (i.e., the order induced by the suffix array), it is neither possible to produce larger cycles.
\end{proof}
}

\block{Time Analysis}
We insert at most $n$~values into the heap.
No value is inserted again.
Finally, we use the following lemma to get a running time of \Oh{n \lg n} :
\begin{lemma}\label{lemDecreaseKey}
The key of a suffix array entry is decreased at most once.
\end{lemma}
\begin{proof}
Let us denote the key of a value~$i$ stored in the heap by~$K[i]$.
Assume that we have decreased the key~$K[j]$ of some value~$j$ stored in the heap 
after we have substituted a substring~$T[i..i+\ell-1]$ with a reference.
It holds that
$K[j] = \SA[i]-\SA[j]-1 > \SA[i]-\SA[j]-1-m \ge K[ {\ISA[ {\SA[j]+m} ]} ]$ for all~$m$ with $1 \le m \le K[j]$, 
i.e., there is no suffix array entry that can decrease the key of~$j$ again.
\end{proof}

\begin{figure}[t]
\newcommand{\XA}[1]{{\color{solarizedRed}#1}}
\newcommand{\X}[1]{{\color{solarizedGreen}\noindent\fbox{#1}}}
\newcommand{\XD}[1]{{\color{solarizedBlue}#1}}
	\centering{\begin{tabular}{l*{17}{c}}
$i$   & 1  & 2  & 3 & 4 & 5 & 6  & 7 & 8  & 9 & 10 & 11 & 12 & 13 & 14 & 15 & 16 & 17
			\\\toprule
			$T$   & {\tt a}  & {\tt a}  & {\tt a} & {\tt b} & {\tt a} & {\tt b}  & {\tt a} & {\tt a}  & {\tt a} & {\tt b}  & {\tt a}  & {\tt a}  & {\tt b}  & {\tt a}  & {\tt b}  & {\tt a}  & {\tt \$}
\\
{\SA[i]}             & 17 & 16 & 7      & 1      & 8      & 11    & 2     & 14     & 5      & 9      & 12     & 3      & 15 & 6      & 10     & 13     & 4
\\
			\midrule
{\LCP[i]}            & -  & 0  & 1      & 5      & 2      & 4     & 6     & 1      & 3      & 4      & 3      & 5      & 0  & 2      & 3      & 2      & 4
\\
{$\LCP^{1}[i]$}      & -  & 0  & \XA{0} & \XD{1} & 2      & 4     & \X{0} & 1      & \XA{0} & 4      & 3      & \XA{0} & 0  & \XA{0} & 3      & 2      & \XA{0}
\\
{$\LCP^{2}[i]$}      & -  & 0  & 0      & 1      & 2      & \X{0} & 0     & \XA{0} & 0      & \XD{2} & \XA{0} & 0      & 0  & 0      & \XD{1} & \XA{0} & 0
\\
{$\LCP^{3}[i]$}      & -  & 0  & 0      & 1      & \XD{1} & 0     & 0     & 0      & 0      & \X{0}  & 0      & 0      & 0  & 0      & \XA{0}  & 0      & 0
		\end{tabular}
	}
	\caption{Step-by-step computation of the \lcpcomp{} compression scheme in \Cref{figParsingArrowsLCPcomp}.
We scan for the largest LCP value in \LCP{} and overwrite values in \LCP{} instead of using a heap.
Each row $\LCP^{i}[i]$ shows the \LCP{} array after computing a substitution.
The LCP value of the starting position of the selected largest repeated substring has a green border.
The updated values are colored, either due to deletion (red) or key reduction (blue).
Ties are broken arbitrarily.
The number of red zeros in each row is equal to the number above the green bordered zero in the corresponding row minus one.
	}
	\label{figElaboratedLcpcompExample}
\end{figure}

\subsubsection{Decompression}
Decompressing \lcpcomp{}-compressed data is harder than decompressing LZ77,
since references in \lcpcomp{} can refer to positions that have not yet been decoded.
\Cref{figParsingArrows} depicts the references built on our running example by arrows.

In order to cope with this problem, we add, 
for each position~$i$ of the original text,
a list~$L_i$ storing the text positions waiting for
this text position getting decompressed.

First, we determine the original text size (the compressor stores it as a VByte before the output of the factorization).
Subsequently, while there is some compressed input, we do the following, using a counting variable~$i$ as a cursor in the text that we are going to rebuild:
\begin{itemize}
	\item If the input is a character~$c$, we write~$T[i] \gets c$, and increment~$i$ by one.
	\item If the input is a reference consisting of a position~$s$ and a length~$\ell$, 
		we check whether $T[s+j]$ is already decoded, for each $j$ with $0 \le j \le \ell-1$:
		\sitemize{\item If it is, then we can restore $T[i+j] \gets T[s+j]$.
			\item Otherwise, we add $i+j$ to the list~$L_{s+j}$.
	}In either case, we increment~$i$ by~$\ell$.
\end{itemize}

	An additional procedure is needed to restore the text completely by processing the lists:
	On writing~$T[i] \gets c$ for some text position~$i$ and some character~$c$, 
	we further write $T[j] \gets T[i]$ for each~$j$ stored in $L_i$ (if $L_j$ is not empty, we proceed recursively).
	Afterwards, we can delete $L_i$ since it will be no longer needed.
	The decompression runs in \Oh{n} time, since we perform a linear scan over the decompressed text, and each text position is visited at most twice.

\JO{\subsubsection{Related Work} 
Rather than adapting the LZ77 compression scheme, 
\cite{NakamuraIBFTS09} and \cite{RistovK15}
compute a grammar by subsequently substituting a longest substring in the text occurring at least twice without overlapping.
After computing such a substring, they exchange each occurrence of it with a new non-terminal.
Compressing the text with a grammar has two main characteristics:
On the one hand, it substitutes \emph{every} non-overlapping occurrence of a substring with the \emph{same} new non-terminal.
On the other hand, it enforces that every substring that is going to be substituted has at least two occurrences \emph{without} overlapping.
}

\subsubsection{Implementation Improvements}\label{seclcpcompImprovement}
In this section, we present an \Oh{n} time compression algorithm alternative to the heap strategy and a practical improvement of the decompression strategy.

\block{Compression}
This strategy computes an array~$A_\ell$ storing all suffix array entries~$j$ with $\LCP[j] = \ell$, for each~$\ell$ with $\varT \le \ell \le \max_k \LCP[k]$.
To compute the references, we sequentially scan the arrays in decreasing order,
	starting with the array that stores the suffixes with the maximum \LCP{} value.
	On substituting a substring~$T[\SA[i]..\SA[i]+\LCP[i]-1]$ with the reference~$(\SA[i-1],\LCP[i])$, 
	we update the \LCP{} array (instead of updating the keys in the heap).
	We set $\LCP[ {\ISA[ {\SA[i]+k} ]} ] \gets 0$ for every $1 \le k \le \LCP[i]-1$ (deletion), and
	$\LCP[j] \gets \min\tuple{\LCP[j], \SA[i]-\SA[j]}$ for every $j$ with $\ISA[ {\SA[i]-\LCP[i]} ] \le j < i$ (decrease key).
	Unlike the heap implementation, we do \emph{not} delete an entry from the arrays.
	Instead, we look up the current \LCP{} value of an element when we process it:
	Assume that we want to process~$A_\ell[i]$.
	If $\LCP[A_\ell[i]] = \ell$, then we proceed as above. 
	Otherwise, we have updated the \LCP{} value of the suffix starting at position $A_\ell[i]$ to the value $\ell' := \LCP[A_\ell[i]] < \ell$.
	In this case, we append $A_\ell[i]$ to $A_{\ell'}$ (if $\ell' < \varT$, we do nothing), and skip computing the reference for $A_\ell[i]$.
	By doing so, we either omit the substring~$A_\ell[i]$ if $\ell' < \varT$, or delay the processing of the value $A_\ell[i]$.
	A suffix array entry gets delayed at most once, analogously to Lemma~\ref{lemDecreaseKey}. In total, the algorithm runs in \Oh{n} time, since it performs basic arithmetic operations on each text position at most twice.

	\block{Decompression}
	We use a heuristic to improve the memory usage.
	The heuristic defers the creation of the lists~$L_i$ storing the 
	text positions that are waiting for the position~$i$ to get decompressed.
	If a reference needs a substring that has not yet been decompressed, we store the reference in a list~$L$.
	By doing so, we have reconstructed at least all substrings that have not been substituted by a reference during the compression.
	Subsequently, we try to decompress each reference stored in~$L$, removing successfully decompressed references from~$L$.
	If we repeat this step, more and more text positions can become restored.
	Clearly, after at most $n$ iterations, we would have restored the original text completely, but this would cost us \Oh{n^2} time.
	Instead, we run this algorithm only for a fixed number of times~$\varA$.
	Afterwards, we mark all not yet decompressed positions in a bit vector~\bv{}, and build a rank data structure on top of~\bv{}.
	Next, we create a list~$L_i$ for each marked text position~$\bv{}.\rank(i)$ as in the original algorithm.
	The difference to the original algorithm is that $L_i$ now corresponds to $\bv{}.\rank(i)$. 
	Finally, we run the original algorithm using the lists~$L_i$ to restore the remaining characters.

	\subsection{LZ78U}\label{secLZ78U}

\begin{figure}[t]
	\floatbox[{\capbeside\thisfloatsetup{capbesideposition={right},capbesidewidth=6.75cm}}]{figure}[\FBwidth]
	{\caption{Dictionary trees of LZ78 and LZ78U. 
		LZ78 factorizes our running example into
		\RubyReset{}\RubyCount{a}$\mid$\RubyCount{aa}$\mid$\RubyCount{b}$\mid$\RubyCount{ab}$\mid$\RubyCount{aaa}$\mid$\RubyCount{ba}$\mid$\RubyCount{aba}$\mid$\RubyCount{ba\$},
		where the vertical bars separate the factors.
		The LZ78 factorization is output as tuples:
		{\tt (0,a)}{\tt (1,a)}{\tt (0,b)}{\tt (1,b)}{\tt (2,a)}{\tt (3,a)}{\tt (4,a)}{\tt (6,\$)}.
		This output is represented by the left trie (a).
		The LZ78U factorization of the same text is
		\RubyReset{}\RubyCount{a}$\mid$\RubyCount{aa}$\mid$\RubyCount{ba}$\mid$\RubyCount{baa}$\mid$\RubyCount{aba}$\mid$\RubyCount{ababa}$\mid$\RubyCount{\$}.
		We output it as
		{\tt (0,a)}{\tt (1,a)}{\tt (0,ba)}{\tt (3,a)}{\tt (1,ba)}{\tt (5,ba)}{\tt (0,\$)}.
		This output induces the right tree (b).
	}
	\label{figParsingTries}
}{\centering{

\begin{minipage}[b]{.4\linewidth}
	\subcaption{LZ78-Trie}\label{figParsingTriesLZ78U}
\tikzstyle{labelnode}=[inner sep=0pt,minimum size=0cm,sloped,above,midway,draw=none,font=\ttfamily]
\begin{forest}
for tree={circle,draw, l sep=20pt}
[0	
	[1, edge label={node[labelnode] {a}} 
	[2, edge label={node[labelnode] {a}} 
	[5, edge label={node[labelnode] {a}} 
	]
]
[4, edge label={node[labelnode] {b}} 
[7, edge label={node[labelnode] {a}} 
		]]]
		[3, edge label={node[labelnode] {b}} 
		[6, edge label={node[labelnode] {a}} 
		[8, edge label={node[labelnode] {\$}} 
		]]]]
\end{forest}
\end{minipage}
	\hspace{2em}
\begin{minipage}[b]{.4\linewidth}
\subcaption{LZ78U-Tree}\label{figParsingTries78}
\tikzstyle{labelnode}=[inner sep=0pt,minimum size=0cm,sloped,above,midway,draw=none,font=\ttfamily]
\begin{forest}
for tree={circle,draw, l sep=20pt}
[0	
[7, edge label={node[labelnode] {\$}}  ]
	[1, edge label={node[labelnode] {a}} 
	[2, edge label={node[labelnode] {a}} 
		]
		[5, edge label={node[labelnode] {ba}} 
	[6, edge label={node[labelnode] {ba}} 
]
		]]
		[3, edge label={node[labelnode] {ba}} 
		[4, edge label={node[labelnode] {a}} 
]
]]
\end{forest}
\end{minipage}
	}}\end{figure}

	A factorization $F_1\cdots F_z = T$ is called the \intWort{LZ78U factorization} of $T$ iff 
$F_x := T[i..j+\ell]$ with $T[i..j] = \argmax_{S \in \menge{F_y : y < x} \cup \menge{\epsilon} } \abs{S}$ and
\[
\ell :=
\begin{cases}
	1 \text{~if~} T[i..j+1] \text{~is a unique substring of~} T, \text{otherwise:} \\
1 + \max \menge{\ell \in \N_0 \mid  \forall k = 1,\ldots,\ell \;\nexists c \in \Sigma \setminus \menge{T[j+k+1]} : T[i..j+k]c \text{~occurs in~} T},
\end{cases}
\]
for all $1 \le x \le z$.
Informally, we enlarge an LZ78 factor representing a repeated substring~$T[i..i+\ell-1]$ to $T[i..i+\ell]$ as long as 
the number of occurrences of $T[i..i+\ell-1]$ and $T[i..i+\ell]$ are the same.

Having the LZ78U factorization $F_1,\ldots,F_z$ of~$T$,
we can output each factor $F_x$ as a tuple $(y,S_x)$ such that $F_x = F_y S_x$, 
where $F_y$ ($0 \le y < x$) is the longest previous factor (set $F_0 := \epsilon$) 
that is a prefix of $F_x$, and $S_x$ is the suffix determined by the factorization. 
We call $y$ the \intWort{referred index} and $S_x$ the \intWort{factor label} of the $x$-th factor.
Transforming the factors to this output induces a dictionary tree, called the \intWort{LZ78U-tree}, in which
\begin{itemize}
	\item every node corresponds to a factor, 
	\item the parent of a node~$v$ corresponds to the referred index of~$v$, and 
	\item the edge between the node of the $x$-th factor and its parent is labeled with the factor label of the $x$-th factor.
\end{itemize}
\Cref{figParsingTries} shows a comparison to the LZ78-trie.
By the definition of the factorizations, the LZ78-trie is a subtree of the suffix \emph{trie},
whereas the LZ78U-tree\JO{\footnote{We named the algorithm LZ78U because each factor label represents a {\bf u}nary path of the suffix trie, unless the path leads to leaf.}} is a subtree of the suffix \emph{tree}. 
The latter can be seen by the fact that the suffix tree compacts the unary paths of the suffix trie.
\JO[This fact is the foundation of the algorithm we present in the following.
It builds the LZ78U-tree on top of the suffix tree.
The algorithm is an easier computable variant of the LZ78 algorithms in~\cite{lzciss,lzcics}.
]{This fact is the foundation of two algorithms we will subsequently show.
Both algorithms build the LZ78U-tree on top of the suffix tree.
They are easier computable variants of the LZ78 algorithms in~\cite{lzciss,lzcics}.
We present
\begin{enumerate}[(1)]
	\item a streaming algorithm with $n \lg n + \abs{\ST{}}$ bits of working space, and \label{it78UA}
	\item an offline algorithm  using $\abs{\ST} + n + z (\lg (2n) + \lg z) + 2z\lg n + \oh{n}$ bits of working space. \label{it78UB}
\end{enumerate}
}

\JO[\block{The Algorithm}]{\block{(\ref{it78UA}) Streaming Algorithm}}
The internal suffix tree nodes can be mapped to the pre-order numbers $[1..n]$ injectively by using rank/select data structures on the suffix tree topology.
This allows us to use $n \lg n$ bits for storing a factor id in each internal suffix tree node.
To this end, we create an array~$R$ of $n \lg n$ bits.
All elements of the array are initially set to zero.
In order to compute the factorization, we scan the text from left to right.
Given that we are at text position~$i$, we locate the suffix tree leaf~$\ell \gets \leafselect[i]$ corresponding to the $i$-th suffix.
Let $p \gets \parent[\ell]$ be $\ell$'s parent. 
\begin{itemize}
	\item If $R[p] \not= 0$, then $p$ corresponds to a factor~$F_x$.
		Let $c$ be the first character of the edge label $\lambda(p,\ell)$.
		The substring $F_x c$ occurs exactly once in~$T$, otherwise $\ell$ would not be a leaf.
		Consequently, we output a factor consisting of the referred index $R[p]$ and the string label~$c$.
		We further increment $i$ by the string depth of $p$ plus one.
	\item Otherwise, using level ancestor queries, we search for the highest node~$v \gets \levelanc[\ell,d]$ with $R[v] = 0$ 
		on the path between the root (exclusively) and~$p$ (iterate over the depth~$d$ starting with zero).
		We set $R[v] \gets z+1$, where $z$ is the current number of computed factors.
		We output the referred index $R[\parent[v]]$ and the string $\lambda(\parent[v],v)$. 
		Finally, we increment $i$ by the string depth of~$v$.
\end{itemize}
Since level ancestor queries can be answered in constant time, we can compute a factor in time linear to its length.
Summing over all factors we get linear time overall.
\JO[We use $n \lg n + \abs{\ST{}}$ bits of working space.]{}
\JO{The code that computes the LZ78U factorization with variant (1) is shown in the appendix.}

	\JO{\block{(\ref{it78UB}) Offline Algorithm}}
	\J[\section{LZ78U Offline Algorithm}]{Instead of directly constructing the array~$R$ that is necessary to determine the referred indices,
we create a list~$F$ storing the marked LZ-trie nodes, and a bit vector~\bv{} marking the internal nodes belonging to the LZ-tree.
Initially, only the root node is marked in \bv{}.
Let $i$, $p$ and $\ell$ be defined as in the above tree traversal.
If $\bv{}[p]$ is set, then we append $\ell$ to $F$ and increment~$i$ by one.
Otherwise, by using level ancestor queries, we search for the highest node~$v$ with $\bv{}[v] = 0$ on the path between the root and~$p$.
We set $\bv{}[v] \gets 1$, and append $v$ to $F$.
Additionally, we increment~$i$ by $\abs{\lambda(\parent[v],v)}$.
By doing so, we have computed the factorization.

In order to generate the final output, we augment $\bv{}$ with a rank data structure, 
and create a permutation~$N$ that maps a marked suffix tree node to the factor it belongs.
The permutation~$N$ is represented as an array of $z \lg z$ bits,
where $N[\bv{}.\rank[1](F[x])] \gets x$, for $1 \le x \le z$.
At this point, we no longer need~$F$.
The rest of the algorithm sorts the factors in the factor index order.
To this end, we create an array~$R$ with $z \lg z$ bits to store the referred indices, and an array~$S$ with $z \lg n$ bits to store the factor labels.
To compute~$S$ and~$R$, we scan all marked nodes in~\bv{}:
Since the $x$-th marked node $v$ corresponds to the $N[x]$-th factor, 
we can fill up~$S$ easily:
If $v$ is a leaf, we store the first character of $\lambda(\parent[v],v)$ in $S[N[x]]$; otherwise ($v$ is an internal node), we store the whole string.
Filling~$R$ is also easy if~$v$ is a child of the root: we simply store the referred index~$0$.
Otherwise, the parent~$p$ of~$v$ is not the root; $p$ corresponds to the $y$-th factor, where $y := N[\bv{}.\rank[1](p)]$.

\JO{We get linear running time with the same argument as for (\ref{it78UA}).}
}\J[The algorithm using $\abs{\ST} + n + z (\lg (2n) + \lg z) + 2z\lg n + \oh{n}$ bits of working space, and runs in linear time.]{}

\block{Improved Compression Ratio}
To achieve an improved compression ratio, we factorize the factor labels:
If~$S_x$ is the label of the $x$-th factor~$f_x$,
then we factorize $S_x = G_1 \cdots G_m$ with
$G_j := \argmax_{S \in \menge{F_y : y < x, \abs{F_y} \ge \varT}\cup\Sigma}\abs{S}$ greedily chosen for ascending values of $j$ with $1 \le j \le m$, with a threshold~$\varT \ge 1$.
By doing so, the string~$S_x$ gets partitioned into characters and former factors longer than $\varT$.
The factorization of~$S_x$ is done in~$\Oh{\abs{S_x}}$ time by traversing the suffix tree with level ancestor queries, 
as above (the only difference is that we do not introduce a new factor to the LZ78U factorization).

\begin{table}[t]
	\centerline{\begin{tabular}{l*{12}{r}}
		\toprule
			collection             & $\sigma$  & $\max$ lcp    & $\textup{avg}_{\LCP{}}$ & bwt-runs      & $z$          & $\max_x \abs{F_x}$ & $H_0$      &  $H_3$
\\\midrule
\textsc{hashtag        }           & \num{179} & \num{54075}   & \num{84}                & \num{63014}K  & \num{13721}K & \num{54056}        & \num{4.59} &  \num{2.46} \\
\textsc{pc-dblp.xml    }           & \num{97}  & \num{1084}    & \num{44}                & \num{29585}K  & \num{7035}K  & \num{1060}         & \num{5.26} &  \num{1.43} \\
\textsc{pc-dna         }           & \num{17}  & \num{97979}   & \num{60}                & \num{128863}K & \num{13970}K & \num{97966}        & \num{1.97} &  \num{1.92} \\
\textsc{pc-english     }           & \num{226} & \num{987770}  & \num{9390}              & \num{72032}K  & \num{13971}K & \num{987766}       & \num{4.52} &  \num{2.42} \\
\textsc{pc-proteins    }           & \num{26}  & \num{45704}   & \num{278}               & \num{108459}K & \num{20875}K & \num{45703}        & \num{4.20} &  \num{4.07} \\
\textsc{pcr-cere       }           & \num{6}   & \num{175655}  & \num{3541}              & \num{10422}K  & \num{1447}K  & \num{175643}       & \num{2.19} &  \num{1.80} \\
\textsc{pcr-einstein.en}           & \num{125} & \num{935920}  & \num{45983}             & \num{153}K    & \num{496}K   & \num{906995}       & \num{4.92} &  \num{1.63} \\
\textsc{pcr-kernel     }           & \num{161} & \num{2755550} & \num{149872}            & \num{2718}K   & \num{775}K   & \num{2755550}      & \num{5.38} &  \num{2.05} \\
\textsc{pcr-para       }           & \num{6}   & \num{72544}   & \num{2268}              & \num{13576}K  & \num{1927}K  & \num{70680}        & \num{2.12} &  \num{1.87} \\
\textsc{pc-sources     }           & \num{231} & \num{307871}  & \num{373}               & \num{47651}K  & \num{11542}K & \num{307871}       & \num{5.47} &  \num{2.34} \\
\textsc{tagme          }           & \num{206} & \num{1281}    & \num{26}                & \num{65195}K  & \num{13841}K & \num{1279}         & \num{4.90} &  \num{2.60} \\
\textsc{wiki-all-vital }           & \num{205} & \num{8607}    & \num{15}                & \num{80609}K  & \num{16274}K & \num{8607}         & \num{4.56} &  \num{2.45} \\
\textsc{commoncrawl}               & \num{115} & \num{246266}  & \num{1727}              & \num{45899}K  & \num{10791}K & \num{246266}       & \num{5.37} &  \num{2.78}
			\\\bottomrule
		\end{tabular}
	}\caption{Datasets of size 200MiB. The alphabet size $\sigma$ includes the terminating {\tt \$}-character. 
		The expression $\textup{avg}_{\LCP{}}$ is the average of all LCP values.
		$z$ is the number of LZ77 factors with $\varT=1$.
		The number of runs consisting of one character in \BWT{} is called bwt-runs.
		$H_k$ denotes the $k$-th order empirical entropy.
}
	\label{tableDatasets}
\end{table}

\section{Practical Evaluation}\label{secEvaluation}

\begin{figure}[t]
	\centering{\href{http://tudocomp.org/charter/?example=sea2017.json}{\includegraphics[width=\textwidth]{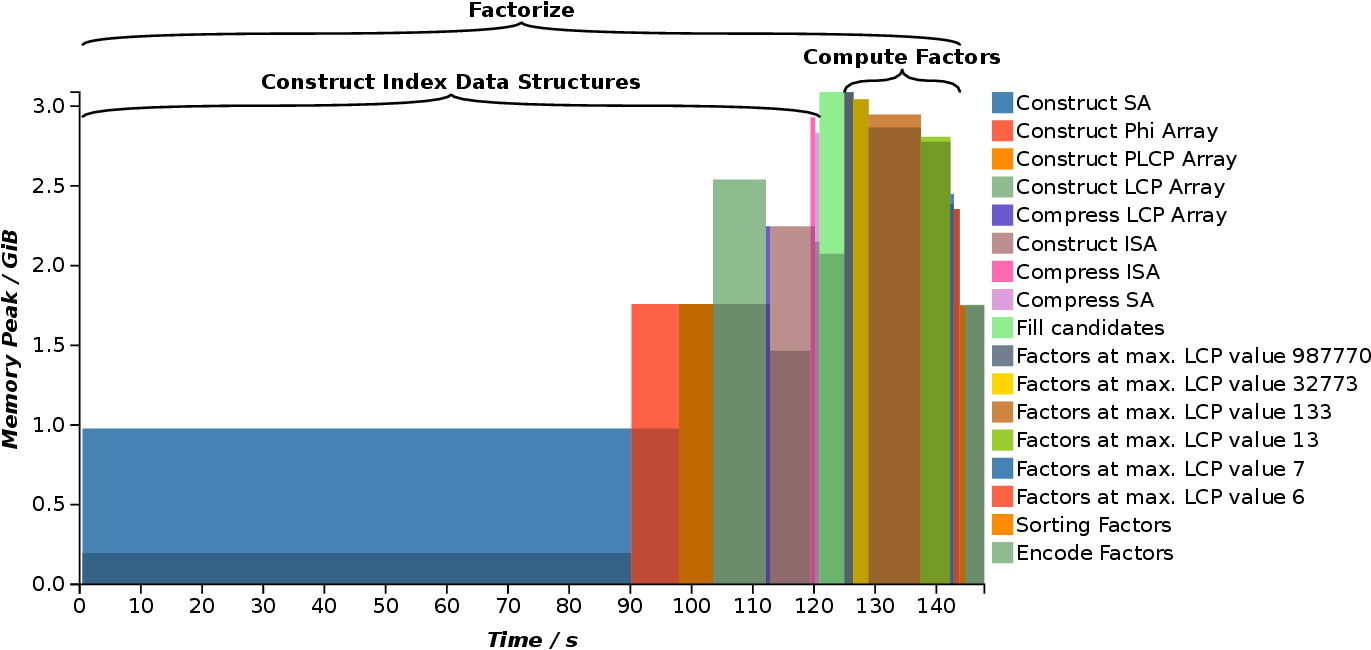}}
}
	\caption{Compression of the collection \textsc{pc-english} with {\tt lcpcomp(coder=sle, threshold=5, comp=arrays)}.
		\SA{} and \LCP{} are built in {\tt delayed} mode.
		Each phase of the algorithm (like the construction of \SA{}) is depicted as a bar in the diagram.
		Each bar is additionally highlighted in a different color with a light and a dark shade.
		The darker part of a phase's bar is the amount of memory already reserved when entering the phase; 
		the lighter part shows the memory peak on top of the already reserved space of the current phase.
		The memory consumption of a phase on termination is equal to the darker bar of the next phase.
		Coherent phases are grouped together by curly braces on the top.
}
\label{figEnglishDiagram}
\end{figure}

\Cref{tableDatasets} shows the text collections used for the evaluation in the tudocomp benchmarks.
We provide a tool that automatically downloads and prepares a superset of the collections used in this evaluation. 
The collections with the prefixes \textsc{pc} or \textsc{pcr} belong to the Pizza\&Chili Corpus\Footnote{PizzaChili}{\url{http://pizzachili.dcc.uchile.cl}}. The Pizza\&Chili Corpus is divided in a real text corpus (\textsc{pc}), and in a repetitive corpus (\textsc{pcr}).
The collection \textsc{hashtag} is a tab-separated values file with five columns (integer values, a hashtag and a title)~\cite{FerraginaPS15}.
The collection \textsc{tagme} is a list of Wikipedia fragments\footnote{\url{http://acube.di.unipi.it/tagme-dataset}}. Finally, we present two new text collections.
The first collection, called \textsc{wiki-all-vital}, consists of the approx.~\num{10000} most vital Wikipedia articles\footnote{\url{https://en.wikipedia.org/wiki/Wikipedia:Vital_articles/Expanded}}. We gathered all articles and processed them with the Wikipedia extractor of TANL~\cite{SIMI16.992} to convert each article into plain text.
The second collection, named \textsc{commoncrawl}, is composed of a random subset of a web crawl\footnote{\url{http://commoncrawl.org}};
this subset contains only the plain texts (i.e., without header and HTML tags) of web sites with ASCII characters.

\block{Setup}
The experiments were conducted on a machine with 32 GB of RAM, an
Intel Xeon CPU~\texttt{E3-1271~v3} and a Samsung SSD {\tt 850 EVO 250GB}.
The operating system was a 64-bit version of Ubuntu Linux~\num{14.04} with the kernel version~3.13.
We used a single execution thread for the experiments.
The source code was compiled using the GNU compiler {\tt g++~6.2.0} with the compile flags {\tt -O3 -march=native -DNDEBUG}.

\block{\lcpcomp{} Strategies}
For \lcpcomp{} we use the heap strategy and the list decompression strategy described in \Cref{seclcpcomp}.
We call them {\tt heap} and {\tt compact}, respectively. 
The strategies described in \Cref{seclcpcompImprovement} are called {\tt arrays} (compression) and {\tt scan} (decompression).
The decompression strategy {\tt scan} takes the number of scans~$\varA$ as an argument.
We encode the remaining substrings of \lcpcomp{} with a static low entropy encoder~{\tt sle}. 
The coder is similar to a Huffman coder, but it additionally treats all $3$-grams of the remaining substrings as symbols of the input.
We evaluated \lcpcomp{} only with the coder~{\tt sle}, since it provided the best compression ratio.
We produced \SA{}, \ISA{} and \LCP{} in the {\tt delayed} mode.

\block{LZ78U Implementation}
We used the suffix tree implementation {\tt cst\_sada} of~SDSL, 
since it provides all required operations like level ancestor queries.

\begin{table}[t]
	{\scriptsize
\begin{verbatim}
pcr_cere.200MB (200.0MiB, sha256=577486b84633ebc71a8ca4af971eaa4e6a91bcddda17f0464ff79038cf928eab)

                       Compressor |     C Time |   C Memory |     C Rate |     D Time |   D Memory | chk |
----------------------------------------------------------------------------------------------------------
                  lz78u(t=5,huff) |     280.2s |     9.2GiB |   12.4643lcpcomp(t=5,heap,compact) |     235.5s |     3.4GiB |    2.8436lcpcomp(t=5,arrays,compact) |     103.1s |     3.2GiB |    2.8505lcpcomp(t=5,arrays,scans(b=25)) |     104.6s |     3.2GiB |    2.8505lzss_lcp(t=5,bit) |      98.5s |     2.9GiB |    4.0530code2 |      16.4s |   230.6MiB |   28.4704huff |       2.7s |   230.5MiB |   28.1072lzw |      14.3s |   480.9MiB |   23.4411lz78 |      13.6s |   480.8MiB |   29.1033bwtzip |      83.6s |     1.7GiB |    6.8688gzip -1 |       2.6s |     6.6MiB |   30.7312gzip -9 |     107.6s |     6.6MiB |   26.2159bzip2 -1 |      13.1s |     9.3MiB |   25.3806bzip2 -9 |      13.8s |    15.4MiB |   25.2368lzma -1 |      12.6s |    27.2MiB |   27.6205lzma -9 |     138.6s |   691.7MiB |    1.9047\end{verbatim}
\vspace{-1em}
}
\caption{Output of the comparison tool for the collection \textsc{pcr-cere}. {\tt C} and {\tt D} denote the compression and decompression phase, respectively.
{\tt b} and {\tt t} are the parameters $\varA$ and $\varT$, respectively.
The tool checks at the last column whether the {\tt sha256}-checksum of the decompressed output matches the input file.}
\label{tableComparetool}
\end{table}

\Cref{figEnglishDiagram} visualizes the execution of \lcpcomp{} with the strategy {\tt arrays} in different phases 
for the collection \textsc{pc-english}.
The figure is generated with the JSON output of {\tt tdc} by the chart visualization application on our website
\url{http://tudocomp.org/charter}.
We loaded the text (200MiB), constructed \SA{} (800MiB, 32 bits per entry), computed \LCP{} (500MiB, 20-bits per entry), computed \ISA{} (700MiB, 28 bits per entry),
and shrunk \SA{} to 700MiB.
Summing these memory sizes gives a memory offset of 1.9GiB when \lcpcomp{} started its actual factorization.
The factorization is divided in LCP value ranges. 
After the factorization, the factors were sorted and finally transformed to a binary bit sequence by {\tt sle}.
Most of the running time was spent on building \SA{}, roughly 1GiB was spent for creating the lists~$L_i$ containing the suffix array entries with an LCP value of~$i$.

Finally, we compare the implemented algorithms of tudocomp with some classic compression programs like {\tt gzip} 
by our comparison tool {\tt compare.py}. The output of the tool is shown in \Cref{tableComparetool}.
The compressor {\tt lzss\_lcp} computes the LZ77 factorization (Def.~\ref{defLZ77}) by a variant of~\cite{Karkkainen2013LTL}.
The compressor {\tt bwtzip} is an alias for the compression pipeline {\tt bwt:rle:mtf:encode(huff)} devised in \Cref{secExampleBWT}.
The programs {\tt bzip2} and {\tt gzip} do not compress the highly repetitive collection \textsc{pcr-cere} as well as 
any of the tudocomp compressors (excluding the plain usage of a coder).
Still, our algorithms are inferior to {\tt lzma -9} in the compression ratio and the decompression speed.
The high memory consumption of {\tt LZ78U} is mainly due to the usage of the compressed suffix tree.

\section{Conclusions}
The framework tudocomp consists of a compression library, the command line executable~{\tt tdc}, a comparison tool,
and a visualization tool.
The library provides classic compressors and standard coders to facilitate building a compressor, or constructing a complex compression pipeline.
Since the library was built with a focus on high modularity, a compression pipeline does not have to get statically compiled.
Instead, the tool {\tt tdc} can assemble a compression pipeline at runtime.
Such a pipeline, given as a parameter to {\tt tdc}, can be adjusted in detail at runtime.

We demonstrated tudocomp's capabilities with the implementation of two new compressors:
\lcpcomp{}, a variant of LZ77, and 
LZ78U, a variant of LZ78.
Both new variants show better compression ratios than their respective originals, but have a higher memory consumption and also slower decompression times. Further research is needed to address these issues.

\block{Future Research}
The memory footprint of \lcpcomp{} could be dropped by exchanging the array implementations of \SA{}, \ISA{} and \LCP{}
with compressed data structures like a compressed suffix array, an inverse suffix array sampling, and a permuted LCP (PLCP) array, respectively.
We are currently investigating a variant that only observes the peaks in the PLCP array to compute the same output as \lcpcomp{}.
If the number of peaks is~$\pi$, then this algorithm needs at most $\pi \lg n$ bits on top of \SA{}, \ISA{} and the PLCP array.

We are optimistic that we can improve the compression ratio of our algorithms 
by adapting sophisticated approaches in how the factors are chosen~\cite{brotli,FarruggiaFV14,FerraginaNV13} and how the factors are finally coded~\cite{DudaTGD15}.

\newpage
\bibliographystyle{plain}

\FloatBarrier
\newpage
\appendix
\Jvar{}
\begin{table}[t]
	\centerline{\begin{tabular}{l*{8}{r}}
			\toprule
			& \multicolumn{5}{c}{compression} & \multicolumn{3}{c}{decompression} \\ \cmidrule(lr){2-6} \cmidrule(lr){7-9}
			collection               & $\varT$ & \#factors      & ratio   & memory       & time      & $\varA$ & memory        & time \\
	\midrule
			\textsc{hashtag        } & \num{5}     & \num{10088662} & 25.47\% & \num{3179.9} & \num{100} & \num{17} & \num{1726  }  & \num{50}\\
			\textsc{pc-dblp.xml    } & \num{5}     & \num{5547102}  & 14.4 \% & \num{2929.7} & \num{99}  & \num{28} & \num{1993.5}  & \num{65}\\
			\textsc{pc-dna         } & \num{21}    & \num{1091010}  & 26.03\% & \num{2925  } & \num{122} & \num{11} & \num{291.2 }  & \num{8} \\
			\textsc{pc-english     } & \num{5}     & \num{11405635} & 27.66\% & \num{3162  } & \num{123} & \num{25} & \num{792.6 }  & \num{36}\\
			\textsc{pc-proteins    } & \num{10}    & \num{1749917}  & 35.91\% & \num{2900  } & \num{124} & \num{13} & \num{362   }  & \num{11}\\
			\textsc{pcr-cere       } & \num{22}    & \num{236551}   & 2.45 \% & \num{3126  } & \num{113} & \num{6}  & \num{454.2 }  & \num{7}\\
			\textsc{pcr-einstein.en} & \num{8}     & \num{24672}    & 0.1  \% & \num{3288.8} & \num{113} & \num{40} & \num{1777.3}  & \num{47}\\
			\textsc{pcr-kernel     } & \num{6}     & \num{512047}   & 1.51 \% & \num{3356.3} & \num{116} & \num{40} & \num{2129.6}  & \num{37}\\
			\textsc{pcr-para       } & \num{22}    & \num{388195}   & 3.27 \% & \num{3060.8} & \num{117} & \num{6}  & \num{402.3 }  & \num{7}\\
			\textsc{pc-sources     } & \num{5}     & \num{8922703}  & 23.36\% & \num{3271  } & \num{98}  & \num{30} & \num{1019.6}  & \num{36}\\
			\textsc{tagme          } & \num{5}     & \num{10986096} & 27.29\% & \num{2987.7} & \num{113} & \num{25} & \num{985.4 }  & \num{41}\\
			\textsc{wiki-all-vital } & \num{5}     & \num{13338470} & 32.46\% & \num{3163  } & \num{117} & \num{27} & \num{870.4 }  & \num{45}\\
			\textsc{commoncrawl }    & \num{4}     & \num{8402041}  & 21.49\% & \num{3254.6} & \num{101} & \num{36} & \num{1206.11} & \num{41}\\
			\bottomrule
		\end{tabular}
	}\caption{Compression and decompression with the \lcpcomp{} strategies {\tt arrays} and {\tt scan}, for fixed parameters $\varT$ and $\varA$. 
		For each collection we chose the $\varT$ with the best compression ratio.
		Having $\varT$ fixed, we chose the $\varA \le 40$ with the shortest decompression running time.
	}
	\label{tableBestCompression}
\end{table}

\FloatBarrier
\section{LZ78U Code Snippet}
\begin{figure}[H]
\begin{lstlisting}
void factorize(TextDS<>& T, SuffixTree& ST, std::function<void(size_t begin, size_t end, size_t ref)> output){
 typedef SuffixTree::node_type node_t;
 sdsl::int_vector<> R(ST.internal_nodes,0,bits_for(T.size() * bits_for(ST.cst.csa.sigma) / bits_for(T.size())));
 size_t pos = 0, z = 0;
 while(pos < T.size() - 1) {
  node_t l = ST.select_leaf(ST.cst.csa.isa[pos]);
  size_t leaflabel = pos;
  if(ST.parent(l) == ST.root || R[ST.nid(ST.parent(l))] != 0) {
   size_t parent_strdepth = ST.str_depth(ST.parent(l));
   output(pos + parent_strdepth, pos + parent_strdepth + 1, R[ST.nid(ST.parent(l))]);
   pos += parent_strdepth+1;
   ++z;
   continue;
  }
  size_t d = 1;
  node_t parent = ST.root;
  node_t node = ST.level_anc(l, d);
  while(R[ST.nid(node)] != 0) {
   parent = node;
   node = ST.level_anc(l, ++d);
  }
  pos += ST.str_depth(parent);
  size_t begin = leaflabel + ST.str_depth(parent);
  size_t end = leaflabel + ST.str_depth(node);
  output(begin, end, R[ST.nid(ST.parent(node))]);
  R[ST.nid(node)] = ++z;
  pos += end - begin;
 }
}
\end{lstlisting}
	\caption{Implementation of the LZ78U algorithm streaming the output}
\end{figure}

\begin{table}[h]
	\centerline{\begin{tabular}{l*{3}{r}l*{2}{r}}
	\toprule
	& \multicolumn{3}{c}{compression} & \multicolumn{3}{c}{decompression} \\ \cmidrule(lr){2-4} \cmidrule(lr){5-7}
	compressor                                                                                                             & memory           & output size     & time          & \multicolumn{1}{c}{strategy}                  & memory           & time\\
	\midrule
	\multicolumn{5}{l}{\textbf{external programs}}\\
	{\tt gzip -1  }                                                                                                        &6.6     & 61.3  & \num{2.19}    &                           & 6.6    & \num{1.045}\\
	{\tt bzip2 -1 }                                                                                                        &9.3     & 55.4  & \num{14.455}  &                           & 8.6    & \num{4.7}\\
	{\tt lzma  -1 }                                                                                                        &27.2    & 46.7  & \num{9.395}   &                           & 19.7   & \num{2.37}\\
	{\tt gzip -9  }                                                                                                        &6.6     & 53.4  & \num{6.86}    &                           & 6.6    & \num{0.97}\\
	{\tt bzip2 -9 }                                                                                                        &15.4    & 50.7  & \num{14.78}   &                           & 11.7   & \num{4.955}\\
	{\tt lzma  -9e}                                                                                                        &691.7   & 29.4  & \num{104.375} &                           & 82.7   & \num{1.56}\\
\multicolumn{5}{l}{\textbf{tudocomp algorithms}}\\
{\tt encode(sle)}                                                                                                        &265.2   & 137.7 & \num{24.145}  &                           & 30.6   & \num{10.095} \\
{\tt encode(huff) }                                                                                                        &230.4   & 135   & \num{5.7}     &                           & 30.4   & \num{9.045}  \\
{\tt bwtzip       }                                                                                                        &1730.6  & 43.7  & \num{83.035}  &                           & 1575   & \num{21.44}  \\
{\tt lcpcomp($\varT=5$,heap)}                                                                                          &3598.9  & 44.1  & \num{228.055} & {\tt compact            } & 6592.2 & \num{33.24}  \\
{\tt lcpcomp($\varT=22$,heap)}                                                                                         &3161.7  & 58.5  & \num{175.21}  & {\tt compact            } & 3981.2 & \num{14.065} \\
\multirow{3}{*}{{\tt lcpcomp($\varT=5$,arrays)} $\left. \vphantom{\begin{tabular}{c}3\\3\\3\end{tabular}}\right\{$ }
                                                                                                                           &        &       &               & {\tt scan($\varA={6}$) } & 4930   & \num{43.1}   \\
                                                                                                                           &3354.2  & 44.3  & \num{107.34}  & {\tt scan($\varA={25}$)} & 2584.5 & \num{33.995} \\
                                                                                                                           &        &       &               & {\tt scan($\varA={60}$)} & 1164.8 & \num{38.925} \\
\multirow{3}{*}{{\tt lcpcomp($\varT=22$,arrays)} $\left. \vphantom{\begin{tabular}{c}3\\3\\3\end{tabular}}\right\{$ }
                                                                                                                           &        &       &               & {\tt scan($\varA={6}$)}  & 1308   & \num{10.925} \\
												                                                                           &2980.6  & 58.5  & \num{109.245} & {\tt scan($\varA={25}$)} & 520.9  & \num{11.265} \\
			                                                                                                               &        &       &               & {\tt scan($\varA={60}$)} & 368.7  & \num{15.635} \\
{\tt lzss(bit)                   }                                                                                         &2980.4  & 60.2  & \num{108.59}  &                           & 230.6  & \num{6.045}  \\
{\tt lz78(bit)                  }                                                                                          &480.8   & 83.1  & \num{17.96}   &                           & 254.9  & \num{11.46}  \\
{\tt lzw(bit)                         }                                                                                    &480.8   & 70.3  & \num{18.97}   &                           & 663.1  & \num{7.05}   \\
	\bottomrule
\end{tabular}
	}\caption{Evaluation of external compression programs and algorithms of the tudocomp framework on the collection \textsc{commoncrawl}.
	}
	\label{tablecommoncrawl}
\end{table}
\section{More Evaluation}
In this section, the execution time is measured in second, and all data sizes are measured in mebibytes (MiB).
In~\Cref{tableBestCompression}, we selected the $\varT$ with the best compression ratio and the $\varA$ with the shortest decompression time.
Although $\varT$ and $\varA$ tend to correlate with the compression speed and decompression memory, respectively,
selecting values for $\varT$ and $\varA$ that yield a good compression ratio or a fast decompression speed seems difficult.

In \Cref{tablecommoncrawl}, we fixed two values of $\varT$ and three values of $\varA$.
The compression ratio of the strategies {\tt heap} and {\tt arrays} differ slightly,
since the \lcpcomp{} compression scheme does not specify a tie breaking rule for choosing a longest repeated substring.

\Cref{figNumFactors} compares the number of factors of {\tt lzss\_lcp} with {\tt lcpcomp}'s {\tt arrays} strategy on all aforementioned datasets.
We varied the threshold~$\varT$ from 4 up to 22 and measured for each~$\varT$ the number of created factors.
In all cases, \lcpcomp{} produces less factors than {\tt lzss\_lcp} with the same threshold.

\begin{figure}[h]
	\centering{\includegraphics[width=0.65\textwidth]{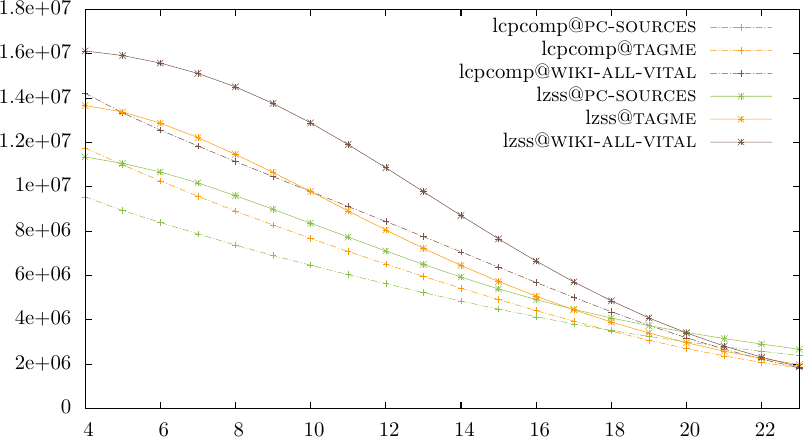}
		\includegraphics[width=0.65\textwidth]{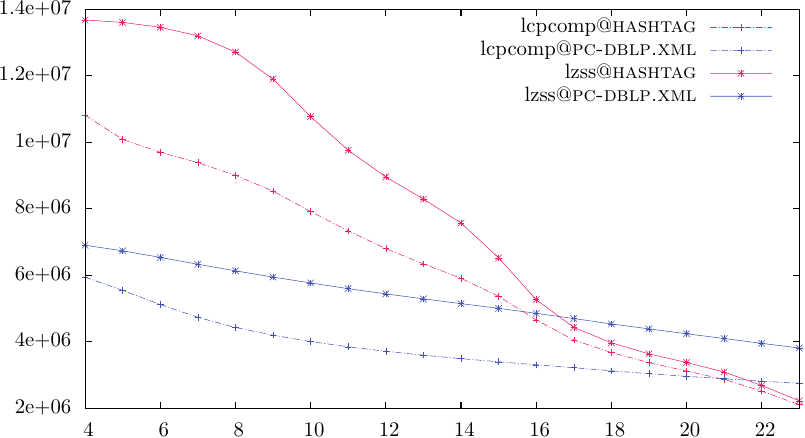}
		\includegraphics[width=0.65\textwidth]{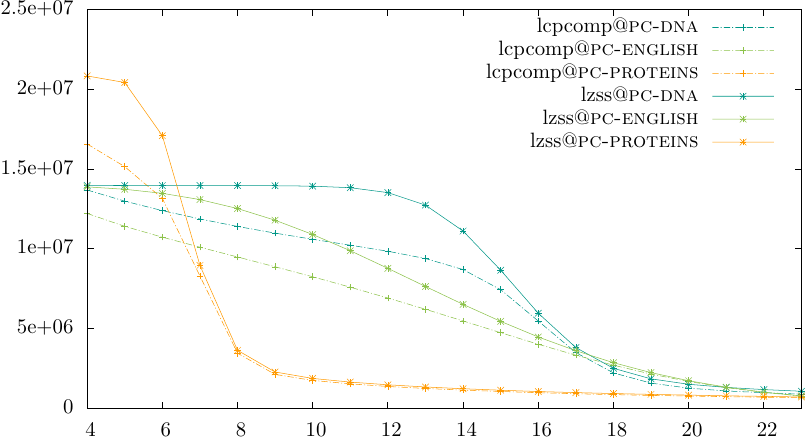}
		\includegraphics[width=0.65\textwidth]{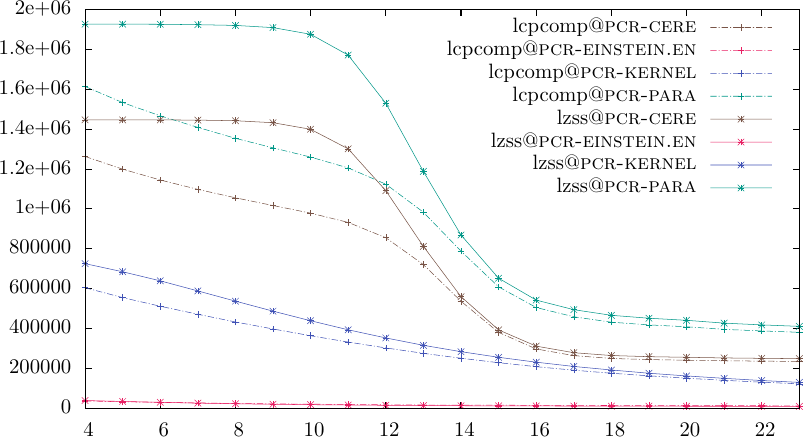}
	}
	\caption{Number of factors ($y$-axis) of \lcpcomp{} and LZ77 on varying the given threshold~$\varT$ ($x$-axis).}
	\label{figNumFactors}
\end{figure}

\FloatBarrier
\newpage

\section{LZ78U Pseudo Codes}

\begin{algorithm}[H]
	\footnotesize
\DontPrintSemicolon $\ST \gets $ suffix tree of $T$ \;
	$R \gets $ array of size~$n$ \tcp*{maps internal suffix tree nodes to LZ trie ids} 
	initialize R with zeros \;
	$pos \gets 1$ \tcp*{text position}
	$z \gets 0$ \tcp*{number of factors}
	\While{$pos \le \abs{T}$}{$\ell \gets \leafselect[ { \ISA[pos] } ] $ \;
\If{$R[\parent[\ell]] \neq 0$ \textbf{\textup{or}} $\parent[\ell] = \textup{root}$ }{output the first character of $\lambda(\parent[\ell],\ell)$ \;
			output referred index $R[\parent[node]]$ \;
			$z \gets z + 1$ \;
			$pos \gets pos + \strdepth[parent] + 1$ \;
		}
		\Else{$d \gets 1$ \tcp*{the current depth}
			\While{$R[\levelanc[\ell,d]] \not= 0$}{$d \gets d + 1$ \;
				$pos \gets pos + \abs{\lambda(\levelanc[\ell,d-1], \levelanc[\ell,d])}$ \;
			}
			$node \gets \levelanc[\ell,d]$ \;
			$z \gets z + 1$ \;
			$R[node] \gets z$ \;
			output string $\lambda(\parent[node],node)$ \;
			output referred index $R[\parent[node]]$ \;
			$pos \gets pos + \abs{\lambda(\parent[node], node)}$ \;
		}
	}
	\caption{Streaming LZ78U}
\label{algo78STSpeed}
\end{algorithm}
\begin{algorithm}[H]
	\footnotesize
\DontPrintSemicolon $\ST{} \gets $ suffix tree of $T$ \; 
$pos \gets 1$ \;
$\bv{} \gets $ bit vector of size $n$ \tcp*{marking the ST nodes belonging to the LZ-trie}
$F \gets $ list of integers \tcp*{storing the LZ-trie nodes in the order when they got explored}
$node \gets \text{root of~} \ST$ \;
\While{$pos \le \abs{T}$}{$node \gets \child[node, {T[pos]} ]$ \tcp*{use \levelanc{} to get \Oh{1} time}
	$pos \gets pos + (\isleaf[node] \text{~?~} 1 : \abs{\lambda(\parent[node], node)}$ \;
	\If{$\isleaf[node] $ \textup{\textbf{or}} $\bv{}[node] = 0$}{$\bv{}[node] \gets 1$ \;
		$F.\text{append}(node)$ \;
		$node \gets \text{root of~} \ST$ \;
	}
}
$\text{add\_rank\_support}(\bv{}) $ \;

$N \gets $ array of length $z$ \tcp*{stores for each marked ST node to which factor it belongs}
\lFor{$1 \le x \le z$}{$N[\bv{}.\rank[1](F[x])] \gets x$ 
}
$F \gets $ integer array of size $z$ \tcp*{storing the referred indices}
$S \gets $ string array of size $z$  \tcp*{storing the string of each factor}

\For{$1 \le x \le z$}{$node \gets \bv{}.\rank[1](x)$ \;
	\lIf{\isleaf[node]}{$S[N[x]] \gets $ first character of $\lambda(\parent[node], node)$
	}
	\lElse{$S[N[x]] \gets \lambda(\parent[node], node)$ 
	}
	\lIf{\parent[node] = \text{root}}{$F[N[x]] \gets 0$ 
	}
	\lElse{$F[N[x]] \gets N[\bv{}.\rank[1](\parent[node])]$ 
	}
}
\Return{(F,S)}
\caption{Computing LZ78U memory-efficiently}
\label{algo78STMemory}
\end{algorithm}

\end{document}